\documentclass[journal]{IEEEtran}

\ifCLASSINFOpdf
  
\else

\fi

\usepackage{epsfig} % for postscript graphics files
\usepackage{epstopdf}
\usepackage{amsmath}
\usepackage{amsthm}
\usepackage{comment}
\usepackage{algorithm}
\usepackage{algorithmic}
\usepackage{amsthm}
\usepackage{amsfonts}
\usepackage{bm}
\usepackage{color}
\usepackage{cite}
\usepackage{stfloats}
\usepackage{url}

\newboolean{showcomments}
\setboolean{showcomments}{true}
\newcommand{\wang}[1]{\ifthenelse{\boolean{showcomments}}
	{ \textcolor[rgb]{1,0,1}{(ZW:  #1)}}{}}
\newcommand{\fliu}[1]{\ifthenelse{\boolean{showcomments}}
	{ \textcolor{red}{(FL:  #1)}}{}}
\newcommand{\ychen}[1]{\ifthenelse{\boolean{showcomments}}
	{ \textcolor{green}{(YC:  #1)}}{}}
\newcommand{\slow}[1]{\ifthenelse{\boolean{showcomments}}
	{ \textcolor{blue}{(SL:  #1)}}{}}

\theoremstyle{definition}
\newtheorem{theorem}{Theorem}
\newtheorem{lemma}[theorem]{Lemma}

\theoremstyle{definition}
\newtheorem{definition}{Definition}
\newtheorem{remark}{Remark}

\begin{document}

\title{Breaking Diversity Restriction: Distributed Optimal Control of Stand-alone DC Microgrids}

\author{
	\vskip 1em
	{
		Zhaojian Wang, Feng Liu, Ying Chen, Steven Low, \emph{Fellow}, \emph{IEEE}, and Shengwei Mei, \emph{Fellow}, \emph{IEEE}
	}
	
    \thanks{This work was supported  by the National Natural Science Foundation
		of China ( No. 51677100, No. 51377092,  No. 51621065, No. 51477081), Foundation of Chinese Scholarship Council (CSC No. 201506215034)
		, 
		the US National Science Foundation through awards EPCN 1619352, CCF 1637598, CNS 1545096, ARPA-E award DE-AR0000699, and Skoltech through Collaboration Agreement 1075-MRA. (Corresponding author: Feng Liu)
	}       % <-this % stops a space
	\thanks{Z. Wang, F. Liu, Y. Chen and S. Mei are with the Department
		of Electrical Engineering, Tsinghua University, Beijing,
		China, 100084 e-mail: (lfeng@tsinghua.edu.cn).}% <-this % stops a space
	\thanks{S. H. Low is with the Department
		of Electrical Engineering, California Institute of Technology, Pasadena, CA, USA, 91105 e-mail:(slow@caltech.edu)}
}

\markboth{Journal of \LaTeX\ Class Files,~Vol.~xx, No.~xx, xx~xxxx}%
{}

\maketitle

\begin{abstract}
Stand-alone direct current (DC) microgrids may belong to different owners and adopt various control strategies. This brings great challenge to its optimal operation due to the difficulty of implementing a unified control. This paper addresses the distributed optimal control of DC microgrids, which intends to break the restriction of diversity to some extent. Firstly, we formulate the optimal power flow (OPF) problem of stand-alone DC microgrids as an exact second order cone program (SOCP) and prove the uniqueness of the optimal solution. Then a dynamic solving algorithm based on primal-dual decomposition method is proposed, the convergence of which is proved theoretically as well as the optimality of its equilibrium point. It should be stressed that the algorithm can provide control commands for the three types of microgrids: (i) power control, (ii) voltage control and (iii) droop control. This implies that each microgrid does not need to change its original control strategy in practice, which is less influenced by the diversity of microgrids. Moreover, the control commands for power controlled and voltage controlled microgrids satisfy generation limits and voltage limits in both transient process and steady state. Finally, a six-microgrid DC system based on the microgrid benchmark is adopted to validate the effectiveness and plug-n-play property of our designs.
\end{abstract}

\begin{IEEEkeywords}
Distributed control, DC microgrid, optimal power flow, diversity restriction, transient constraint.
\end{IEEEkeywords}

%\markboth{IEEE TRANSACTIONS ON }
{}
%\IEEEpeerreviewmaketitle

\section{Introduction}

%\subsection{Motivation}

\IEEEPARstart{M}{icrogrids} are clusters of distributed generators (DGs), energy storage systems (ESSs) and loads, which are generally categorized into two types: alternating current (AC) and direct current (DC) microgrids \cite{Justo:AC, Planas:AC}. In the past decade, research has been concentrated on enhancing the performance of AC microgrids. However, some generations and loads are inherently DC, such as photo-voltaic (PV), battery, computer and electrical vehicle (EV) \cite{Dragicevic:DC1, Dragicevic:DC2, Elsayed:DC}. DC microgrids more naturally integrate them and can eliminate unnecessary conversion processes, which improves system efficiency and reliability. In addition, DC systems do not face problems such as reactive power compensation, frequency stability and synchronization \cite{Dragicevic:DC1}, which makes it more and more popular in power system. 
In DC microgrids, hierarchical control is often utilized \cite{Guerrero:Hierarchical, Shafiee:Hierarchical}, i.e., primary control, secondary control and tertiary control, which can be implemented in either a centralized manner or a distributed manner. In the centralized manner, a control center is needed to accumulate information from microgrids, compute command and send it back to them. With the increasing number of microgrids as well as uncertainties of renewable generations and load demands, centralized control faces a great challenge, i.e., it is less and less applicable due to problems, e.g., single point failures, heavy communication burden of control center and lack of ability to respond rapidly enough \cite{Yazdanian:Distributed}. These problems highlight the need for a distributed control strategy that will require no control center and less communication. This paper addresses this need.

In the hierarchical control architecture, the primary control is almost decentralized. The most popular control manner is the droop control \cite{Guerrero:Hierarchical, Maknouninejad:Optimal}, where load sharing is mainly determined by the droop coefficient. As pointed out in \cite{Khorsandi:A, Gu:Mode}, droop control cannot achieve proper load sharing sometimes, especially in systems with unequal resistances and different modes. Many improvements are investigated \cite{Khorsandi:A, Gu:Mode, Chen:Autonomous, Cook:Decentralized, Lu:State}. Taking into consideration the effect of different line impedances, \cite{Khorsandi:A} proposes a decentralized control strategy to achieve perfect power sharing.
In \cite{Gu:Mode}, a mode-adaptive decentralized control strategy is proposed for the power management in DC microgrid, which enlarges the control freedom compared with the conventional droop control. In \cite{Lu:State}, a decentralized method is proposed to adjust the droop coefficient by the state-of-charge of storage, which can achieve equal load sharing. However, similar to the AC power system, primary control in DC system  suffers from voltage deviation in the steady state.

To eliminate the voltage deviation, distributed secondary control is developed. The most widely used method is consensus based control \cite{Olfati-Saber:Consensus}, where there is usually a global control variable, e.g., global voltage deviation, while each agent only has its local estimation. In the DC system, each agent may represent a DG or a microgrid. By exchanging information with neighbors, the value of the variable will be identical for all agents finally \cite{Shafiee:Distributed, Shafiee:Hierarchical, Nasirian:Distributedc, Lv:Discrete}. In \cite{Shafiee:Hierarchical}, each microgrid uses dynamic consensus protocol to estimate the global averaged voltage with the local and neighboring estimation. Then, the estimated voltage is compared with the reference value and fed to a PI controller to eliminate the voltage deviation. This method is further improved in \cite{Nasirian:Distributedc} by adding a current consensus regulator, where the control goal is to achieve globally identical current ratio compared with the rated current of each microgrid. By doing so, the equal load sharing can be obtained. The discrete consensus method is used in \cite{Lv:Discrete} to restore average voltage with accurate load sharing. The consensus based secondary control can realize equality among agents, however, the results may not be optimal. 

Tertiary control is to achieve the optimal operation by controlling the power flow among microgrids or among DGs within a microgrid \cite{Che:DC, Xiao:Hierarchical, Moayedi:Distributed}. Conventionally, tertiary control provides reference operation point for the system. Its time scale is much slower than real time control. However, values of renewable generations and loads may change rapidly due to uncertainties, which makes reference point obtained by tertiary control sub-optimal in the new situation. This requires us to combine real-time coordination and steady-state optimization together, i.e., the optimization solution should be sent to the system in real time. Similar works are given in both AC system \cite{Changhong:Design, Li:Connecting, Dorfler:Breaking} and DC system \cite{Moayedi:Unifying, Hamad:Multiagent, Wang:a}. The critical thought in \cite{Moayedi:Unifying, Hamad:Multiagent, Wang:a} is that the incremental generation cost of each microgrid is identical in the steady state. In \cite{Moayedi:Unifying}, economic dispatch problem is formulated and the incremental generation cost is regarded as consensus variable. Using the consensus method, the incremental generation cost will be identical in the steady state, and optimality is achieved. Similar method is also used in \cite{Wang:a}, where the sub-gradient is added to the consensus approach in order to accelerate the convergence. These works are very inspiring in combination of optimal operation and real time control, but they still have some restrictions. For example, the original control strategy of a microgrid has to be revised to the proposed method, which is hard to apply as microgrids may belong to different owners and adopt various control strategies. These problems highlight the need for a distributed control strategy that is less influenced by the diversity.

In this paper, we investigate the distributed optimal power flow control among stand-alone DC microgrids, which is less influenced by their original control strategies. We construct an optimal power flow (OPF) model for stand-alone microgrids with an exact SOCP relaxation and further prove the uniqueness of its optimal solution. By using the primal-dual decomposition method, a distributed dynamic algorithm is proposed, which provides control commands for different control strategies such as power control, voltage control and droop control. This implies that we do not change the original control schemes of microgrids, which breaks the restriction of microgrid diversity. In addition, constraints of generation capacity limits and voltage limits are enforced even in the transient process of control commands, which implies power commands are always feasible and voltage commands are safe for converters. In this regards, it increases the security of DC system. Furthermore, we also prove the convergence of the algorithm and optimality of the equilibrium point.
The contributions of this paper have following aspects:
\begin{itemize}
	\item The OPF model of stand-alone DC power system is formulated, and the uniqueness of its optimal solution is proved. 
	\item A fully distributed algorithm is proposed to achieve the optimal solution of the OPF problem, where only communications with neighbors are needed with minimal communication burden.
	\item The proposed method does not change the original control strategy of each DG, which adapts to three most common control modes: power control, voltage control and droop control, breaking restriction of microgrid diversity. 
	\item The control commands for power controlled and voltage controlled microgrids satisfy generation limits and voltage limits in both transient process and steady state. 
\end{itemize}

The rest of the paper is organized as follows. In Section II, network model of DC microgrids is introduced. In Section III, OPF model for stand-alone DC microgrids is formulated. In Section IV, the dynamic algorithm is proposed, the optimality and convergence of the algorithm are proved theoretically. The implementation approach is designed in Section V. Case studies are given in Section VI. Finally, Section VII concludes the paper.

%\wang{breaking into two sections, the number of Lemma Assumption}

\section{Network Model}

A stand-alone DC system is composed of a cluster of microgrids connected by lines. Each microgrid is treated as a bus with generation and load. Then the whole system is modeled as a connected graph $\mathcal{G}:=(\mathcal N, \mathcal E)$, where  $\mathcal N=\{1,2,...n\}$ is the set of microgrids and $\mathcal E\subseteq \mathcal N\times \mathcal N$ is the set of lines. If two microgrids $i$ and $k$ are connected by a tie line directly, we denote $(i,k)\in \mathcal E$, and abbreviated by $i\sim k$. The resistance of line $(i,k)$ is $r_{ik}$. The power flow from microgrid $i$ to microgrid $k$ is $P_{ik}$, and the current from microgrids $i$ to $k$ is $I_{ik}$. Let $m:= |\mathcal E|$ be the number of lines. 

For each microgrid $i\in\mathcal N$, let $p_i^g(t)$ denote the generation at time $t$ and $p_i^d$ denotes its constant load demand. Denote the voltage at bus $i$ as $V_i$. DGs in the DC microgrids may have different control strategies, such as power control, voltage control and droop control.  Power control and voltage control only require their reference values, which are not introduced here in detail. Droop control takes the form:
\begin{align}
	v_i-v_i^*=-k_i(p^g_i-\hat p_i).
	\label{droop control}
\end{align}
where $v_i=V_i^2$, $k_i>0$ is the droop coefficient, $v_i^*$ is the voltage square reference, and $k_i, v_i^*$ are constants. $\hat p_i$ is the power when $v_i =v_i^*$, which is a variable in the rest of the paper. 

Denote the current in line $(i,k)$ from $i$ to $k$ as $I_{ik}$, which is defined  
\begin{align}
	I_{ik}=(V_i-V_k)/r_{ik}
\end{align}
Then the power $P_{ik}$ from $i$ to $k$ is 
\begin{align}
   P_{ik}=V_iI_{ik}=V_i(V_i-V_k)/r_{ik}
\end{align}
Consequently, the power balance in one node is
\begin{subequations}
	\begin{align}
		p_i^g - p_i^d &= V_i \sum\nolimits_{k:k \in {N_i}} (V_i-V_k)/r_{ik} \label{power flow1}\\
					  &= \sum\nolimits_{k:k \in {N_i}} {{P_{ik}}} .
		\label{power flow2}
	\end{align}
where $N_i$ is the set of microgrids connected with microgrid $i$ directly. 
\end{subequations}
Our goal is to provide control commands for microgrids adopting different control strategies, which must satisfy the operational constraints:
\begin{subequations}
	\begin{align}
		0 &\le p_i^g \le \overline p_i^g	
		\label{generation constraints}	\\
		\underline V_i &\le V_i \le \overline V_i
		\label{voltage constraints}		
	\end{align}
	where $\overline p_i^g$ is the upper limit of generation in DG $i$, $\underline V_i, \overline V_i$ are lower and upper limits of voltage. For power controlled DG, \eqref{generation constraints} is a hard limit, which must be satisfied  even during the transient process. Otherwise it is non-executable. For voltage controlled DG, (\ref{voltage constraints}) is not a hard limit, but it also should be satisfied during the transient. This is because overlimit voltage is not secure for the converter nor the operator.
	\label{operation constraints}
\end{subequations}

\section{Optimal Power Flow Problem}

\subsection{OPF Model}
Existing OPF models are mainly for grid-connected DC system \cite{Gan:Optimal}. However, DC microgrids also operate in isolated mode in many situations such as in remote areas or islands. In terms of this, we formulate the ordinary OPF in the stand-alone DC power system. 
\begin{subequations}
	\begin{eqnarray}\text{OPF:~}
	\mathop {\min }\limits_{p^g,v,W} & \!\!\!\!\!\! & f=\sum\nolimits_{i \in {\cal N}} {f_i(p_i^g)} 
	\label{initial objective function0} \\
%	\text{over} & \!\!\!\!\!\! & x := (p^g,v,W)\nonumber\\
	\text{s. t.} & \!\!\!\!\!\! & 
	(\text{\ref{generation constraints}})\nonumber\\
	& \!\!\!\!\!\! &{\underline{V}_i}^2 \le {v_i} \le {{\overline V}_i}^2,\quad i \in {\cal N}
	\label{voltage square constraints0}\\
	& \!\!\!\!\!\! &
	p_i^g - p_i^d= \sum\limits_{k:k\sim i} (v_i-W_{ik})/r_{ik}, \ i \in {\cal N}
	\label{power injection}\\
	& \!\!\!\!\!\! & W_{ik}\ge0, \quad i\sim k
	\label{W positive}\\
	& \!\!\!\!\!\! & W_{ik}=W_{ki}, \quad i\sim k
	\label{line W}
	\\
	& \!\!\!\!\!\! &R_{ik} \succeq 0, \quad i\sim k
	\label{R Positive definite}
	\\
	& \!\!\!\!\!\! &\text{rank}(R_{ik})=1, \quad i\sim k
	\label{rank constraints0}
	\end{eqnarray}
	where $R_{ik}=\left[ {\begin{array}{*{20}{c}}
		v_i    &   W_{ik}		\\
		W_{ki} &   v_j
		\end{array}} \right]$.  
	\label{initial optimization problem0}
\end{subequations} If $\text{rank}(R_{ik})=1$ always holds, $W_{ik}$ can be divided into  $W_{ik}=V_iV_k$. 
The cost \eqref{initial objective function0} is a function of generation in each node, which should satisfy 
\begin{itemize}
	\item[\textbf{A1}:] $f_i(p_i^g)$ is strictly increasing when $p_i^g\ge -p_i^d$ for $i\in \cal N$, second order continuously differentiable and strongly convex 	$(f_i^{''}(p_i^g)\ge \alpha >0)$.
\end{itemize}
Constraint (\ref{voltage square constraints0}) is derived from (\ref{voltage constraints}), and (\ref{power injection}) is from (\ref{power flow1}). Constraint (\ref{R Positive definite}) implies matrix $R_{ik}$ is positive semi-definite, and (\ref{rank constraints0}) guarantees the rank of $R_{ik}$ be 1.
The difference between (\ref{initial optimization problem0}) and $\text{OPF}^{'}$ in \cite{Gan:Optimal} is that there is no substation node with fixed voltage in (\ref{initial optimization problem0}). (\ref{initial optimization problem0}) is not convex due to constraint (\ref{rank constraints0}). Remove (\ref{rank constraints0}), and we get the SOCP relaxation of (\ref{initial optimization problem0}). 
\begin{eqnarray} \text{SOCP:~}
	\mathop {\min }\limits_{p^g,v,W} & \!\!\!\!\!\! & \sum\nolimits_{i \in {\cal N}} {f_i(p_i^g)} \nonumber	\\
%	\text{over} & \!\!\!\!\!\! & x := (p^g,v,W)\nonumber\\
	\text{s. t.} & \!\!\!\!\!\! & (\text{\ref{generation constraints}}), (\text{\ref{voltage square constraints0}})- (\text{\ref{R Positive definite}})\nonumber
\end{eqnarray}
It has been proved in \cite{Li:On} that the relaxation is exact provided that: 1) $\overline V_1=\overline V_2=\dots=\overline V_n$; 2) $p_i^d>0$; 3) $\sum\nolimits_{i\in \cal N} (p_i^g - p_i^d)>0$; 4) $f_i(p_i^g)$ is strictly increasing when $p_i^g\ge p_i^d$ for $i\in \cal N$.
%\begin{enumerate}
%	\item $\overline V_1=\overline V_2=\dots=\overline V_n$;
%	\item $p_i^d>0$;
%	\item $\sum\nolimits_{i\in \cal N} (p_i^g - p_i^d)>0$;
%	\item $f_i(p_i^g-p_i^d)$ is strictly increasing for $i\in \cal N$.
%\end{enumerate}
In this paper, conditions 1), 2), 3) are satisfied.  With assumption A1, 4) is also satisfied.

To improve the numerical stability of the \text{SOCP}, we have the following stable SOCP problem. 
\begin{subequations}
	\begin{eqnarray}\text{SSOCP:~}
		\mathop {\min }\limits_{p^g,P,l,v} & \!\!\!\!\!\! & \sum\nolimits_{i \in {\cal N}} {f_i(p_i^g)} 
		\label{initial objective function} \\
%		\text{over} & \!\!\!\!\!\! & x := (p^g,P,l,v)\nonumber\\
		\text{s. t.} & \!\!\!\!\!\! & 
		(\text{\ref{power flow2}}),(\text{\ref{generation constraints}})\nonumber\\
%		v_i-v_i^*=-k_i(p^g_i-\hat p_i)\\
%        & \!\!\!\!\!\! & p_i^g - p_i^d = \sum\limits_{k:k \in {N_i}} {{P_{ik}}} \\
		& \!\!\!\!\!\! & {P_{ik}} + {P_{ki}} = {r_{ik}}{l_{ik}},\quad  i\sim k
		\label{network loss}\\
		& \!\!\!\!\!\! & {v_i} - {v_k} = {r_{ik}}\left( {{P_{ik}} - {P_{ki}}} \right),\ i\sim k
		\label{voltage relation}
		\\
		& \!\!\!\!\!\! &{l_{ik}} \ge {{P_{ik}^2}}/{{{v_i}}},\quad i\sim k
		\label{current relation}
		\\
		& \!\!\!\!\!\! &{\underline{V}_i}^2 \le {v_i} \le {{\overline V}_i}^2,\quad i \in {\cal N}
		\label{voltage square constraints}
	\end{eqnarray}
	where $l_{ik}=|I_{ik}|^2$ are squared line currents, and $l_{ik}=l_{ki}$. Constraint (\ref{current relation}) is the SOCP relaxed form. The detailed explanation of (\ref{network loss})-(\ref{current relation}) is found in \cite{Gan:Optimal}, which is omitted here.
	\label{initial optimization problem}
\end{subequations}

According to Theorem 5 in \cite{Gan:Optimal}, \text{SOCP} and \text{SSOCP} are equivalent, i.e., there exists a one-to-one map between the feasible set of \text{SOCP} and the feasible set of \text{SSOCP}, which is 
\begin{align}
	\label{map}
	P_{ik}&=(v_i-W_{ik})/r_{ik},\quad i\sim k; \nonumber\\
	l_{ik}&=(v_i-W_{ik}-W_{ki}+v_k)/r_{ik}^2,\quad i\sim k.
\end{align}

In some microgrids, droop control is utilized. However, the solution of (\ref{initial optimization problem}) cannot guarantee $v_i$ and $p_i^g$ satisfy (\ref{droop control}), which implies that the optimal solution may not be achieved in reality. In this regard, we add droop control to the constraints, then the problem becomes	
\begin{eqnarray}\text{DSOCP:~}
	\mathop {\min }\limits_{p^g,P,l,v,\hat p} & \!\!\!\!\!\! & \sum\nolimits_{i \in {\cal N}} {f_i(p_i^g)}  \nonumber
	\label{droop objective function} \\
%	\text{over} & \!\!\!\!\!\! & x := (p^g,P,l,v,\hat p)\nonumber\\
	\text{s. t.} & \!\!\!\!\!\! & 
	(\text{\ref{droop control}}), (\text{\ref{power flow2}}), (\text{\ref{generation constraints}}), (\text{\ref{network loss}})- (\text{\ref{voltage square constraints}}) \nonumber
\end{eqnarray}
In \text{DSOCP}, the droop coefficient is a constant, while $\hat p$ is an optimization variable, making DSOCP a convex problem.

To help design the algorithm, an equivalent optimization problem is formulated. 
\begin{eqnarray}\text{ESOCP:~}
	\mathop {\min }\limits_{p^g,P,l,v,\hat p} & \!\!\!\!\!\! & \sum\limits_{i \in {\cal N}} {f_i(p_i^g)}  + \sum\limits_{i \in {\cal N}} {\frac{1}{2}y_i^2} + \sum\limits_{i \in {\cal N}} {\frac{1}{2}z_i^2}
	\label{revised objective function} \nonumber\\
%	\text{over} & \!\!\!\!\!\! & x := (p^g,P,l,v,\hat p)\nonumber\\
	\text{s. t.} & \!\!\!\!\!\! & 
	(\text{\ref{droop control}}), (\text{\ref{power flow2}}), (\text{\ref{generation constraints}}), (\text{\ref{network loss}})- (\text{\ref{voltage square constraints}})\nonumber
\end{eqnarray}
	where $y_i=v_i+k_ip^g_i-v_i^*-k_i\hat p_i$, $z_i={{{ {p_i^g - p_i^d - \sum\nolimits_{k:k \in {N_i}} {{P_{ik}}} } }}}$. 
	
Since for any feasible solution of ESOCP we all have $y_i=z_i=0$, ESOCP is equivalent to DSOCP. $y_i$ and $z_i$ are only put here to accelerate the convergence of the algorithm \cite{Feijer:Stability}. 

\subsection{Uniqueness of Optimal Solution}
Before introducing the results, we give an assumption.
\begin{itemize}
	\item[\textbf{A2}:] The OPF (\ref{initial optimization problem0}) is feasible.
\end{itemize}
Then, we have the following theorem. 
\begin{theorem}
	\label{uniqueness}
	Suppose A1 and A2 hold. The optimal solution of SSOCP is unique.
\end{theorem}
The proof of Theorem \ref{uniqueness} is given in Appendix A.

%of (\ref{initial optimization problem}) and (\ref{droop optimization problem}), we have the following Theorem.
\begin{theorem}
	\label{Th:same solution}
	Denote the optimal solution of SSOCP as $x^{1*}=(p^{g1*},P^{1*},l^{1*},v^{1*})$ and the optimal solution of DSOCP as $x^{2*}=(p^{g2*},P^{2*},l^{2*},v^{2*},\hat p^{2*})$. Then,
	\begin{enumerate}
		\item 	there exists an unique $\hat p^{2*}$ making $(x^{1*},\hat p^{2*})$ the optimal solution of DSOCP;
		\item the optimal solution of DSOCP is unique;
		\item   $(p^{g2*},P^{2*},l^{2*},v^{2*})=(p^{g1*},P^{1*},l^{1*},v^{1*})$.
	\end{enumerate}
\end{theorem}
The proof of Theorem \ref{Th:same solution} is given in Appendix A.

Suppose the optimal solution of DSOCP is $x^{2*}=(p^{g2*},P^{2*},$ $l^{2*},v^{2*},\hat p^{2*})$ with the droop coefficient $k_{i1}$. From the proof of 1) in Theorem \ref{Th:same solution}, if $k_{i1}$ changes to $k_{i2}$, there exists an unique $\hat p^{2*}(k_{i2})$ making $(p^{g2*},P^{2*},$ $l^{2*},v^{2*},\hat p^{2*}(k_{i2}))$ be the optimal solution of DSOCP. This implies that droop coefficient does not influence the optimal solution of SSOCP. 

\begin{lemma}
	\label{same feasible solution}
	Optimization problem ESOCP and DSOCP have identical feasible solutions.
\end{lemma}
It is easy to prove Lemma \ref{same feasible solution} as $y_i=0$ and $z_i=0$ for any feasible solution.

\begin{remark}
	From Theorem \ref{Th:same solution}, it is shown that the unique optimal solution still exists even if droop control is considered, and we can obtain the optimal $\hat p_2^*$ in droop control. Moreover, for these microgrids that do not adopt droop control, ESOCP can provide the optimal output power and voltage references. In addition, for different droop coefficient $k_i$, $(p_2^{g*},P_2^*,l_2^*,v_2^*)$ in the optimal solution does not change. Thus, for microgrids that do not adopt droop control, we can just assign an imaginary droop control to them when formulating ESOCP, i.e., assuming all the microgrids adopt droop control when building ESOCP. This does not influence the optimal solution of these microgrids adopting power control and voltage control. In this regard, our method adapts to three different control strategies.
\end{remark}

\section{Control Scheme Design}

\subsection{Distributed Algorithm}
Based on the primal-dual algorithm, we propose the following distributed approach to solve the ESOCP, which is
\begin{subequations}
	\label{algorithm}
	\begin{align}
		\dot p_i^g &=  \left[ {p_i^g - \left( {G_i(p_i^g) - {\mu _i} +k_i\epsilon_i + z_i+k_iy_i } \right)} \right]_0^{\overline p_i^g} - p_i^g 		
		\label{generation dynamics}\\
		{{\dot v}_i} & = \bigg [v_i - \bigg(y_i+ \sum\limits_{k \in {N_i}} {{\gamma _{ik}}} +\epsilon_i  -  \sum\limits_{k \in {N_i}} {{\rho _{ik}}\frac{{P_{ik}^2}}{{v_i^2}}}  \bigg)\bigg]_{{\underline{V}_i}^2}^{{\overline{V}_i}^2} -v_i 
		\label{voltage square dynamics}\\
		{{\dot P}_{ik}}& =  - \left( {{\mu _i} + {\lambda _{ik}} - {\gamma _{ik}}{r_{ik}} + 2{\rho _{ik}}{{{P_{ik}}}}/{{{v_i}}} - z_i} \right)
		\label{line power dynamics}		\\
		{{\dot l}_{ik}} &=  - \left( { - {\lambda _{ik}}{r_{ik}} - {\rho _{ik}} - {\rho _{ki}}} \right)
		\label{current square dynamics}\\
		\dot{\hat p}_i&=k_i\epsilon_i+k_iy_i \label{hat p dynamics}
		\\
		{{\dot \mu }_i} &=  -  \bigg( {p_i^g - p_i^d - \sum\limits_{k:k \in {N_i}} {{P_{ik}}} } \bigg)\label{mu dynamics}\\
		\dot\epsilon_i &= v_i+k_ip^g_i-v_i^*-k_i\hat p_i\\	
		{{\dot \lambda }_{ik}} &=   {{P_{ik}} + {P_{ki}} - {r_{ik}}{l_{ik}}}\\
		{{\dot \gamma }_{ik}} &=   {{v_i} - {v_k} - {r_{ik}}\left( {{P_{ik}} - {P_{ki}}} \right)}\label{gamma dynamics} \\
		\label{rho dynamics}
		{{\dot \rho }_{ik}} &=   \left[ {{{P_{ik}^2}}/{{{v_i}}} - {l_{ik}}} \right]_{{\rho _{ik}}}^ + 
	\end{align}
\end{subequations}
where  $G_i(p_i^g):=\frac{\partial}{\partial p_i^g}f_i(p_i^{g})$. 
For any $x_i, a_i, b_i$ with $a_i \le b_i$, $[x_i]_{a_i}^{b_i} := \min \{ b_i, \max \{ a_i, x_i \} \}$. Operator $[x_i]^+_{a_i}$  is
\begin{align}
	[x_i]^+_{a_i}&=\left\{
	\begin{array}{ll}
	x_i,& \text{if} \ a_i>0 \ \text{or} \  x_i>0;\\
	0,& \text{otherwise}.
	\end{array}
	\right.
	\nonumber
\end{align}

The algorithm (\ref{algorithm}) is fully distributed where each MG updates its internal states $p_i^g$, $P_{ik}, l_{ik}$, $v_i$, $\hat{p_i}$, $\mu _i$, $\epsilon_i$, $\lambda _{ik}$, $\gamma _{ik}$, $\rho _{ik}$ relying only on local information and neighboring information. The neighboring information only appear in variables $P_{ki},v_k,\rho_{ki},k\in N_i$. 

Next, we will investigate the boundedness of  $(p_i^g(t), v_i(t))$. Firstly we introduce the assumption
{\begin{itemize}
		\item[\textbf{A3}:] The initial states of the dynamic system (\ref{algorithm}) are finite, and $(p_i^g(0), v_i(0))$ satisfy constraint (\ref{operation constraints}).
\end{itemize}}
Define the set 
\begin{align}
X:=\left\{(p_i^g, v_i)|0\le p_i^g\le\overline p_i^g, \underline v_i\le v_i \le \overline v_i \right\},
\label{boundedness set}
\end{align}
then we will prove the boundedness property of  $(p_i^g(t), v_i(t))$.
\begin{lemma}
	Suppose assumption A3 holds. Then constraint (\ref{operation constraints}) is satisfied for all $t \ge 0$, i.e. $(p^g(t),v(t)) \in X$ for all $t \ge 0$ where $X $ is defined in (\ref{boundedness set}).
	\label{lemma transient boundedness}
\end{lemma}
\begin{proof}[Proof of Lemma \ref{lemma transient boundedness}]
	Note that (\ref{generation dynamics}) is an inertia link with input 
	$$u^g_i=\left[ {p_i^g - \left( {G_i(p_i^{g}) - {\mu _i} +k_i\epsilon_i + z_i+k_iy_i } \right)} \right]_0^{\overline p_i^g}$$
	According to the feature of inertia link, $p^g(t)\in X$ for all $t \ge 0$ holds as long as $u^g_i(t)\in X$ for all $t \ge 0$. Thus, we know $p^g(t)\in X$ for all $t \ge 0$ always holds. Similarly, $v(t)\in X$ for all $t \ge 0$ always holds. This completes the proof.
\end{proof}
Lemma \ref{lemma transient boundedness} implies that inequality constraints are enforced even in the transient for $p_i^g$ and $v_i$.

%\section{Optimality and Convergence}
\subsection{Optimality of Equilibrium Point}
In this subsection, we will prove that the equilibrium points of (\ref{algorithm}) are primal-dual optimal for ESOCP and its dual, and vice versa. Firstly, the definition of equilibrium points of (\ref{algorithm}) and the optimal solution of ESOCP are given in Definition 1 and Definition 2 respectively. 

Given $x_p := (p^g, P,  l, v, \hat{p})$, $ x_d := (\mu, \epsilon, \lambda , \gamma , \rho )$, two definitions are introduced.
\begin{definition}
	\label{def:ep.2}
	A point $(x_p^*, x_d^*) := (p^{g*}, P^*, l^*, v^*, \hat{p}^*, \mu^*, \epsilon^*,$ $ \lambda^*,$ $\gamma^* , $ $\rho^* )$ is an equilibrium point of (\ref{algorithm}) if the right-hand side of (\ref{algorithm}) vanishes at $(x_p^*, x_d^*)$.
\end{definition}

\begin{definition}
	A point $(x_p^*, x_d^*)$ is primal-dual optimal if $x^*_p$ is optimal 
	for ESOCP and $x_d^*$ is optimal for its dual problem.
\end{definition}

To prove the optimality of $(x_p^*, x_d^*)$, we make the following assumption:
{\begin{itemize}
	\item[\textbf{A4}:] Slater's condition for ESOCP holds.
\end{itemize}}

We first illustrate that the saturation of controller does not influence the optimal solution of ESOCP, which is introduced in Lemma \ref{lemma saturation optimal solution}.
\begin{lemma}
	\label{lemma saturation optimal solution}
	Suppose A1, A2 and A4 hold. If $(x_p^*, x_d^*)$ is primal-dual optimal, we have 
	\begin{align}
		p_i^{g*} &= \left[ {p_i^{g*} - \left( {G_i(p_i^{g*}) - {\mu _i^*} +k_i\epsilon_i^* + z_i^*+k_iy_i^* } \right)} \right]_0^{\overline p_i^g}  \nonumber\\
		v_i^* &=\bigg [v_i^* - \bigg(y_i^*+ \sum\limits_{k \in {N_i}} {{\gamma _{ik}^*}} +\epsilon_i^*  
		- \sum\limits_{k \in {N_i}} {{\rho^* _{ik}}\frac{{(P_{ik}^*)^2}}{{(v_i^*)^2}}}  \bigg)\bigg]_{{\underline{V}_i}^2}^{{\overline{V}_i}^2}\nonumber
	\end{align}
\end{lemma}
The proof of Lemma \ref{lemma saturation optimal solution} is given in Appendix B.
Based on Lemma \ref{lemma saturation optimal solution}, we have the following Theorem.
\begin{theorem}
	\label{theorem optimality}
	Suppose A1, A2, A3 and A4 hold. A point $(x_p^*, x_d^*)$ is primal-dual optimal if and only if it is an equilibrium of the dynamic system (\ref{algorithm}).
\end{theorem}
The proof of Theorem  \ref{theorem optimality} is given in Appendix B.

\begin{figure*}[!t]
	\centering
	\includegraphics[width=0.92\textwidth]{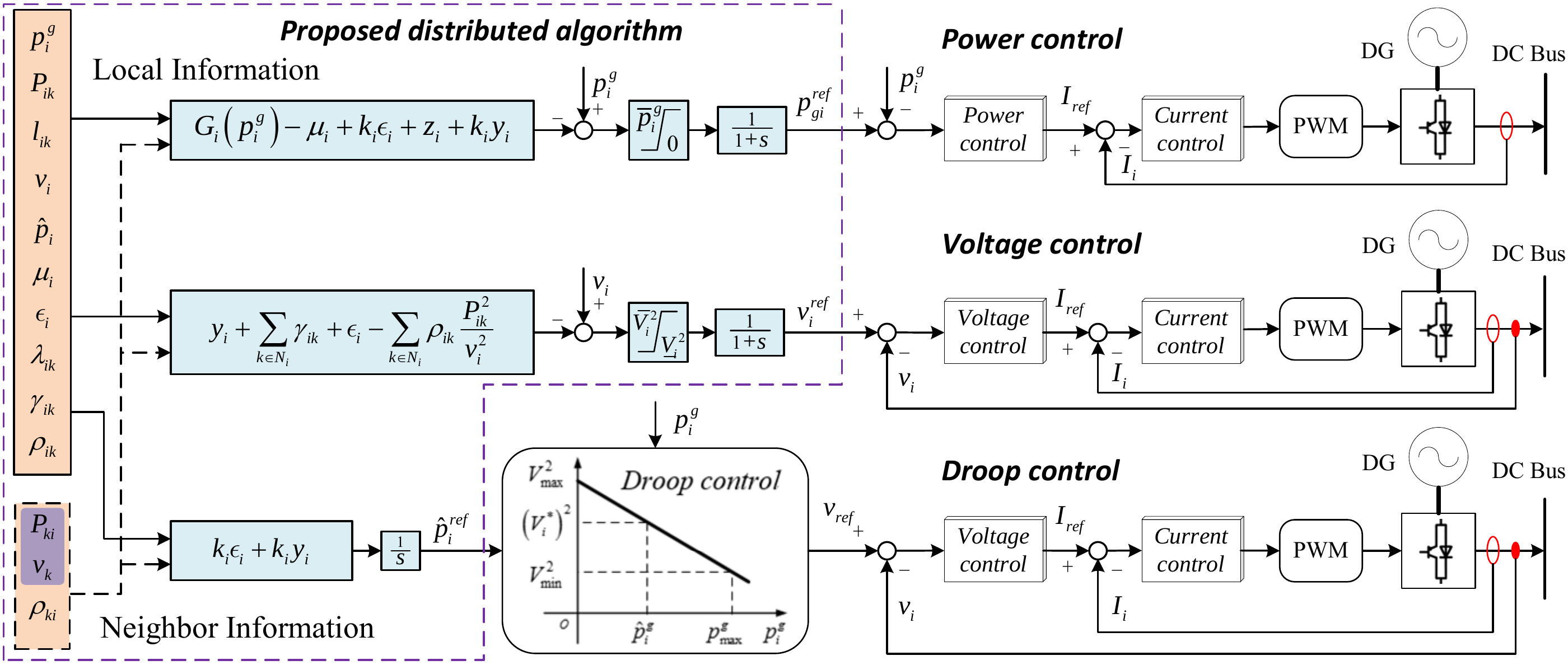}
	\caption{Control diagram of the proposed method}
	\label{control diagram}
\end{figure*}

\subsection{Convergence Analysis}

In this subsection, we will justify the convergence of the algorithm (\ref{algorithm}) by projection gradient theory combined with invariance principle for switched system. 

Define the sets $\sigma_\rho$ 
\begin{align}
\sigma_\rho &:= \left\{(i,k)\in \mathcal{E} \, | \, \rho_{ik}=0,    \, {{{P_{ik}^2}}/{{{v_i}}} - {l_{ik}}}<0
\right\} \nonumber	
\end{align}
Then (\ref{rho dynamics}) is equivalent to
\begin{align}
\dot\rho_{ik} &= \left\{
\begin{array}{ll}
\rho_{ik}({{P_{ik}^2}}/{{{v_i}}} - {l_{ik}}),
& \text{if}\ (i,k) \notin \sigma_\rho ;\\
0,
& \text{if}\ (i,k) \in \sigma_\rho .        \end{array}
\right.
\label{eq:rho}
\end{align}
From (\ref{eq:rho}), it is easy to know $\rho_{ik}(t)\ge0, \forall t$.

Denote $x:=(p_i^g, v_i, P_{ik}, l_{ik}, \hat{p}_i, \mu _i, \epsilon_i, \lambda _{ik}, \gamma _{ik}, \rho _{ik})$ and define $F(w)$ in a  fixed $\sigma_\rho$.
\begin{align}
	F\left( x \right) = \left[ \begin{array}{l}
	 {G_i(p_i^{g}) + k_i\epsilon_i - {\mu _i} + z_i + k_iy_i } \\
	 {y_i+\sum\limits_{k:k \in {N_i}} {{\gamma _{ik}}} +\epsilon_i  - \varphi _i^ -  + \varphi _i^ +  - \sum\limits_{k \in {N_i}} {{\rho _{ik}}\frac{{P_{ik}^2}}{{v_i^2}}} } \\
	 {{\mu _i} + {\lambda _{ik}} - {\gamma _{ik}}{r_{ik}} + 2{\rho _{ik}}\frac{{{P_{ik}}}}{{{v_i}}} - z_i} \\
	 { - {\lambda _{ik}}{r_{ik}} - {\rho _{ik}} - {\rho _{ki}}} \\
	 -k_i\epsilon_i-k_iy_i
	 \\
	 {p_i^g - p_i^d - \sum\limits_{k:k \in {N_i}} {{P_{ik}}} } \\
	- (v_i+k_ip^g_i-v_i^*-k_i\hat p_i)\\
	- ({{P_{ik}} + {P_{ki}} - {r_{ik}}{l_{ik}}}) \\
	- ( {{v_i} - {v_k} - {r_{ik}}\left( {{P_{ik}} - {P_{ki}}} \right)}) \\
	- \left[ {{{P_{ik}^2}}/{{{v_i}}} - {l_{ik}}} \right]_{{\rho _{ik}}}^ + 
	\end{array} \right]
\end{align}
$F(x)$ is continuously differentiable in a fixed $\sigma_\rho$.

We further define the set
%$$S:=X\times R^{2m+m+n+n+n+m+m+2m}$$
$$S:=X\times R^{7m+3n}$$
where $X$ is given in (\ref{boundedness set}). For any $x$, the projection $x-F(x)$ onto $S$ is 
\begin{align}
	{H}\left( x \right) := \text{Proj}_{S}\left( {x - F\left( x \right)} \right):=\arg\min_{y\in S} \| y - (x-F(x)) \|_2 \nonumber
\end{align}
where $\| \cdot \|_2$ is the Euclidean norm. Then, the algorithm (\ref{algorithm}) can be rewritten as 
\begin{align}
	\dot x(t) =  {{H}\left( x(t) \right) - x(t)} 
	\label{projection rewritten dynamics}
\end{align}
A point $x^*\in S$ is an equilibrium of (\ref{projection rewritten dynamics}) if and only if it is a fixed point of the projection:
$${{H}\left( x^* \right) = x^*} $$
Let $E:=\{\ x\in S\ |\ {{H}\left( x(t) \right) - x(t)}=0\ \}$ be the set of equilibrium points.

\begin{theorem}
	\label{theorem convergence}
	Suppose A1, A2, A3 and A4 hold. Then 
	every trajectory $x(t)$ of (\ref{projection rewritten dynamics}) starting from a finite initial state  asymptotically converges to some equilibrium $x^*\in E$ as $t\rightarrow +\infty$ that is optimal for problem ESOCP.	
\end{theorem}
The proof of Theorem \ref{theorem convergence} is provided in  Appendix C.

\section{Implementation}
Each microgrid may adopt different control strategies: voltage control, power control or droop control. Different strategies require different control commands, which are $p_i^g, v_i$, $\hat{p}_i$ respectively. Since we combine optimization with real time control, values of $p_i^g, v_i$, $\hat{p}_i$ in the transient process are also sent to the corresponding DGs as the control commands. To distinguish with state variables $p_i^g, v_i$, $\hat{p}_i$, control commands sent to DGs are denoted as $p_{gi}^{ref}$, $v_i^{ref}$ and $\hat{p}_i^{ref}$ respectively.
For voltage and power control, the algorithm (\ref{algorithm}) can provide $p_{gi}^{ref}$ and $v_i^{ref}$ that are all feasible even in the transient process. For the droop control microgrids, we can supply $\hat{p}_i^{ref}$, which ensures the system operate in the optimal status. The control diagrams for three types of microgrid are shown in Fig.\ref{control diagram}.

In Fig.\ref{control diagram}, the left part is the proposed algorithm, the inputs of which are local information  $p_i^g$, $P_{ik}, l_{ik}$, $v_i$, $\hat{p_i}$ $\mu _i$, $\epsilon_i$, $\lambda _{ik}$, $\gamma _{ik}$, $\rho _{ik}$ and neighbor information $P_{ki}, v_k, \rho_{ki}, k\in N_i$. The outputs are $p_{gi}^{ref}$ , $v_i^{ref}$  and $\hat{p}_i^{ref}$. The right part is the diagrams of three control strategies: power control, voltage control and droop control. For power controlled DG, it has two control loops, power loop and current loop, where $I_{ref}$ is the current reference for the current control loop and $I_i$ is the measured current. For voltage controlled DG, it also has two control loops, power loop and current loop, where $v_{i}$ is the measured voltage. For droop controlled DG, it has three control loops, droop control loop, power loop and current loop, where both voltage and current need to be measured. 
%They receive different control commands, which are $p_{gi}^{ref}$ , $v_i^{ref}$  and $\hat{p}_i^{ref}$ respectively. 

From Fig.\ref{control diagram}, we can see that our method adapts to three commonly used control strategies, which in some sense implies it breaks restriction of various control strategies in microgrids in achieving optimal operation point. 

\begin{remark}
	In fact, for microgrid $i$, neighbor information $P_{ki}$ and $v_k$ can be estimated locally by the following equations
	$$P_{ki}=P_{ik}-r_{ik}I_{ik}^2$$
	$$v_k=(\sqrt{v_i}-r_{ik}I_{ik})^2$$
	where line current $I_{ik}$ from microgrid $i$ to $k$ can be measured locally. Then, only $\rho_{ki}, k\in N_i$ need to be exchanged between neighbors, which implies that the communication burden is minimized.
	
	In the real system, some microgrid may switch off or switch on unexpectedly. The system should also operate optimally in this situation. This requires the controller has the capability of plug-n-play, which will be shown in the simulation. 
\end{remark}

%\cite{Wang:Decentralized}

\section{Case Studies}
\subsection{Test System}

To verify the effectiveness of the proposed approach, a multi-microgrid DC system is utilized, the topology of which is based on the low voltage microgrid benchmark in \cite{Papathanassiou:A}. The system includes two feeders with six dispatchable DGs, which are divided into six microgrids based on corresponding DGs. The Breaker 1 is open, and the system operates in an isolated way. The simulation is performed in PSCAD joint with Matlab. More specifically, the DC microgrids are modeled in PSCAD, whereas algorithm (\ref{algorithm}) is computed in Matlab. They are combined by a user-defined interface. The control commands obtained in Matlab are sent to corresponding microgrids in PSCAD through the interface. Conversely, MATLAB can also collect data from PSCAD. In this regards, the PSCAD represents the physical system, while MATLAB represents the cyber system. Communication exists in cyber system to obtain the neighborhood information. The joint simulation flowchart is illustrated in Fig.\ref{interface}.

\begin{figure}[!t]
	\centering
	\includegraphics[width=0.5\textwidth]{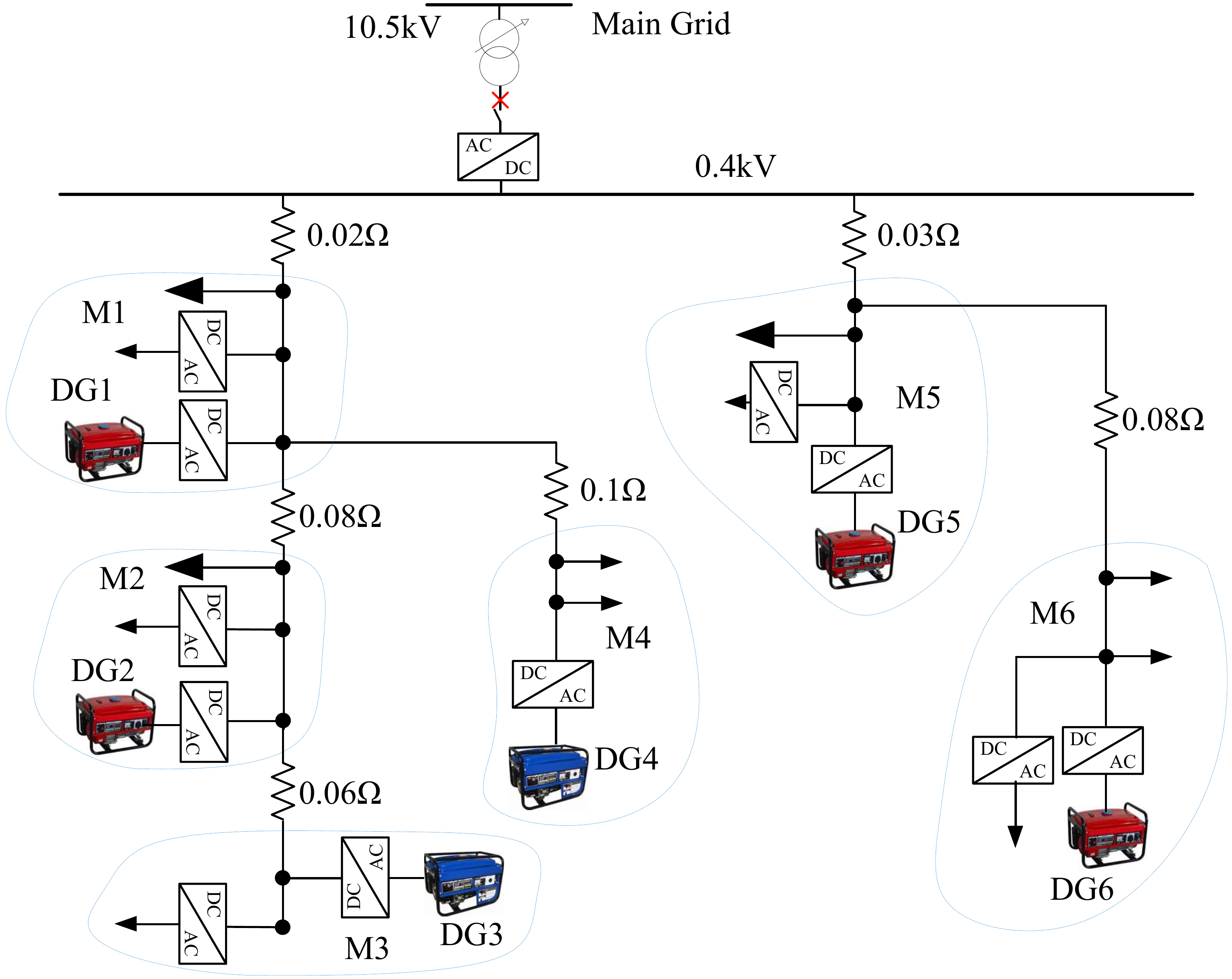}
	\caption{Six microgrids system}
	\label{test system}
\end{figure}

\begin{figure}[!t]
	\centering
	\includegraphics[width=0.4\textwidth]{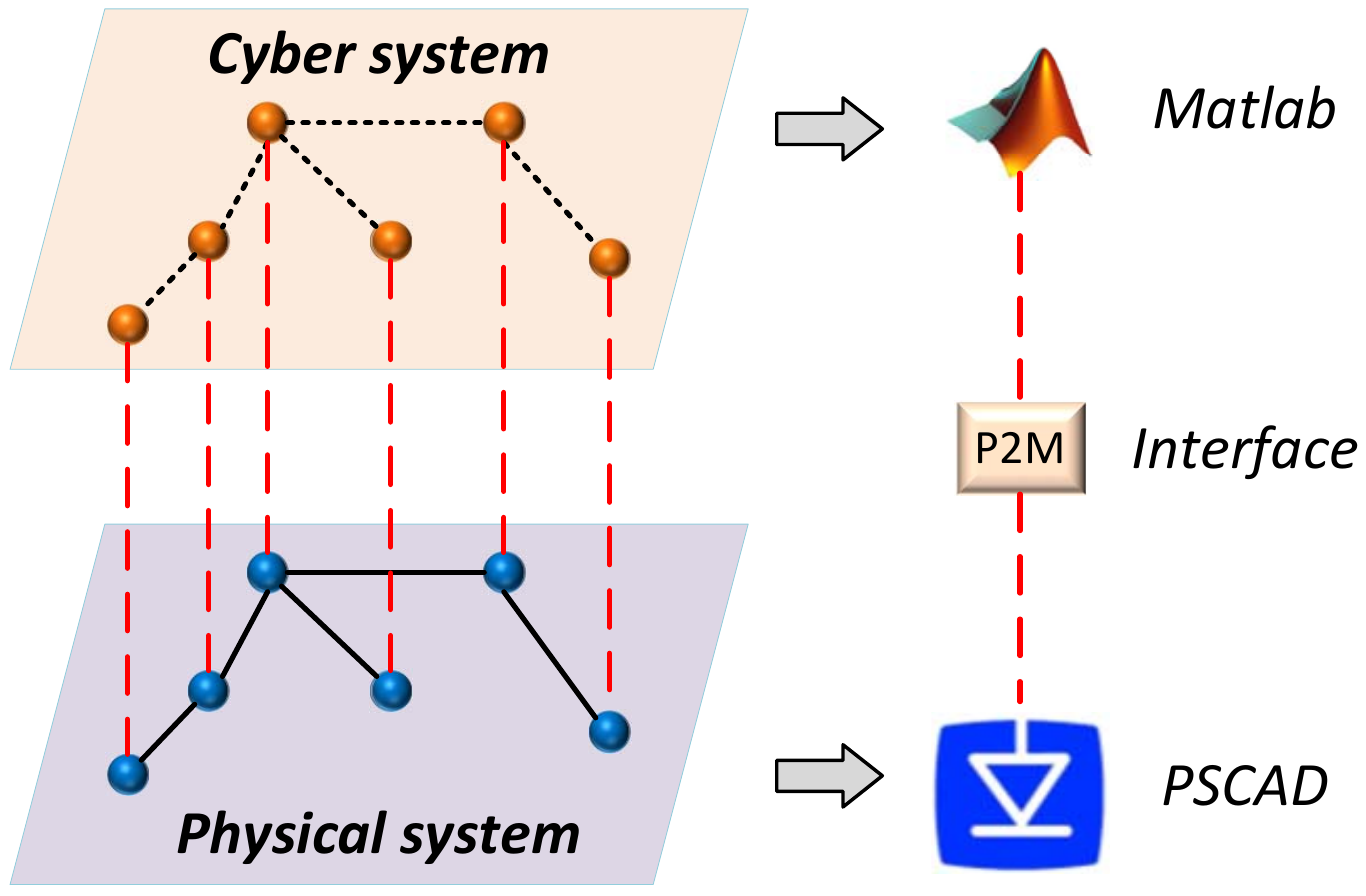}
	\caption{Simulation design combined with PSCAD and MATLAB}
	\label{interface}
\end{figure}

\begin{table}[!t]
	\renewcommand{\arraystretch}{1.3}
	\centering
	\caption{\textsc{System parameters}}
	\label{SysPara}
	\begin{tabular}{c c c c c c c}
		\hline
		\hline 
		  & $M_1$ & $M_2$ & $M_3$ & $M_4$ &$M_5$ &$M_6$\\
		\hline
		$a_i$ & 0.036 & 0.03 & 0.035 & 0.03 &0.035 &0.042\\
		$b_i$ & 1 & 1 & 1 & 1 & 1 & 1\\
		$p_i^d  \text{(kW)}$ & 51 & 50 & 52 & 49 & 52 &50\\
		$\overline{p}_i^g  \text{(kW)}$ & 50 & 60 & 55 & 60 & 55 & 45\\
		$\overline{V}_i  \text{(V)}$ & 420 & 420 & 420 & 420 & 420 & 420\\
		$\underline{V}_i  \text{(V)}$ & 380 & 380 & 380 & 380 & 380 & 380\\				
		$k_i$ & 0.12 & 0.125 & 0.164 & 0.131 & 0.156 & 0.131
		\\
		\hline
		\hline
	\end{tabular}
\end{table}

The objective function is set as $f_i(p_i^g)= {\frac{{{a_i}}}{2}{{\left( p_i^g \right)}^2} + b_ip_i^g }$, which represents the generation cost of the whole system as the generation cost also takes the quadratic form \cite{Xu:Distributed, Wang:a}. If $b_i>0$, it satisfies A1. 
Some parameters for these microgrids are provided in Table.\ref{SysPara} .
The simulation case is that load demands in each microgrid is $(41, 40, 42, 39, 42, 40)\text{kW}$ at first, then they will increase to $(51, 50, 52, 49, 52, 50)\text{kW}$ at time 1s. All the three regular control strategies are utilized, i,e, DG1 and DG6 adopt droop control, DG2 and DG5 use power control while DG3 and DG4 adopt voltage control. 

\subsection{Accuracy Analysis}

In this subsection, we also use CVX tool in Matlab to solve ESOCP, results of which after load increases are utilized as basic values to validate the accuracy of the proposed approach. Results of these two methods are compared in Table \ref{comparison}.

\begin{table}[!t]
	\renewcommand{\arraystretch}{1.3}
	\caption{Comparision with Centralized Optimization \label{comparison}}%%%Table caption 
	\centering
	{%%%%%%%%% \tabcolsep command is to adjust the inter column spacing of table
		\begin{tabular}{c|c|c|c|c|c|c}%%%The number of columns has to be defined here
			\hline
			\hline
			%\toprule [2pt]
			 & \multicolumn{3}{c|}{Generation \text{(kW)}} & \multicolumn{3}{c}{Power reference \text{(kW)}}
			\\ %%%% Table body
			\hline
			$\text{DG}_i$	& ${p}_g^{d}$  & ${p}_g^{c}$ &  $e$ (\%) & $\hat p^{d}$  & $\hat p^{c}$ &  $e$ (\%) 
			\\ %%%% Table body
			\hline
			1&	48.1701&	48.1823&	-0.0253&   46.7315&	 46.8543&	-0.2621\\
			2&  57.4041&	57.4227&    -0.0324&   56.3536&  56.7590&	-0.5381\\
			3&  49.3516&	49.3602&	-0.0174&   48.2000&	 48.6611&	-0.3311\\
			4&  56.9853&	56.9861&    -0.0014&   56.9264&	 56.9861&   -0.1048\\
			5&  49.9761&	50.0053&    -0.0584&   48.8660&	 48.3499&	0.4470\\
			6&  42.1217&	42.1495&    -0.0660&   39.2465&	 39.2176&	0.0737\\
			\hline
			\hline
		\end{tabular}}{}
	\end{table}

In Table \ref{comparison}, ${p}_g^{d}$ and $\hat p^{d}$ are $p_i^g$ and $\hat{p}_i$ of all DGs obtained by the proposed approach, while ${p}_g^{c}$ and ${\hat p}^{c}$ are values obtained using the CVX tool. $e$ is the errors of  ${p}^{d}$ and $\hat p^{d}$ with regards to ${p}^{c}$ and ${\hat p}^{c}$. From results in Table \ref{comparison}, it can be seen that the absolute errors between $p_g^d$ and $p_g^c$ in each MGs are smaller than $0.07\%$. In addition, the absolute errors between $\hat p^d$ and $\hat p^c$ are smaller than $0.6\%$. Both validate the accuracy of the proposed approach.

\subsection{Dynamic Process}
%\subsubsection{All power control}

In this subsection, 
%The tie line dynamics and $\hat{p}$ dynamics for DG1 and DG6 are illustrated in Fig.\ref{Voltage1}. It is shown that tie line powers and $\hat{p}$ will remain steady after 4s.
%
%\begin{figure}[!t]
%	\centering
%	\includegraphics[width=0.49\textwidth]{droop_line.pdf}
%	\caption{Tie line powers (left) and $\hat{p}$ (right) dynamics}
%	\label{Voltage1}
%\end{figure}
we analyze the impacts of generation limits on the dynamic property. To do this, we compare dynamic responses of the inverter outputs in microgrids 2 and 5 by two scenarios: with and without saturation. The trajectories in two cases are given in Fig.\ref{mixed generation}. In both cases, the same steady state generations are achieved. However, with the saturated controller, the generations of DG2 and DG5 remain within the limits in both transient and steady state. On the contrary, generatons of DG2 and DG5 violate their upper limits in the transient, which is practically infeasible.  

Similarly, we also compare the voltage dynamics of DG3 and DG4 in two scenarios: with and without saturation. The trajectories in two cases are given in Fig.\ref{mixed voltage}. In both cases, the same steady state voltages are achieved. However, with the saturated controller, the voltages of DG3 and DG4 remain within the limits in both transient and steady state. On the contrary, voltages of DG3 and DG4 violate their upper limits in the transient process if saturation is not considered. As we know, the high voltage is both harmful to the power electronic equipments and  system operators. In this sense, our method can increase system security.

%\subsubsection{Mixed control}

\begin{figure}[!t]
	\centering
	\includegraphics[width=0.49\textwidth]{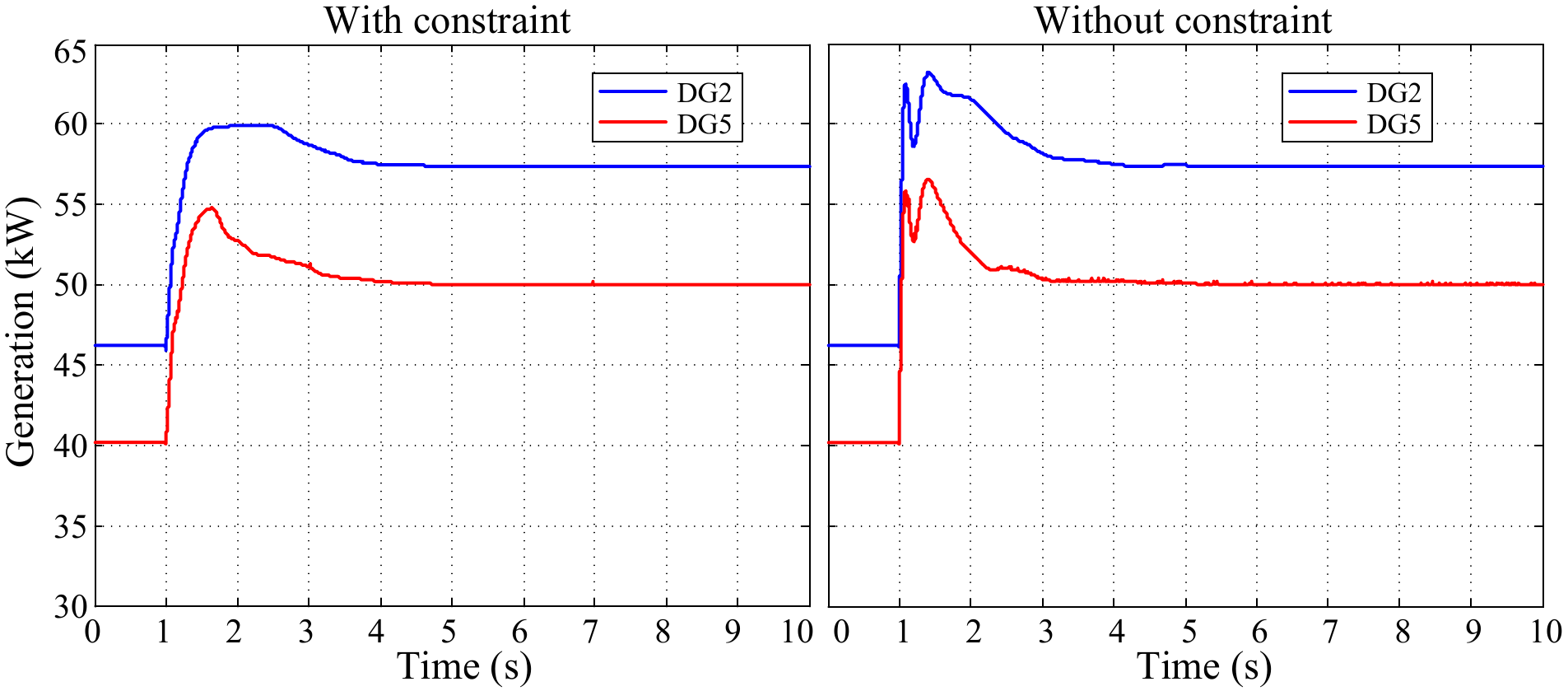}
	\caption{Generation dynamics with and without constraints }
	\label{mixed generation}
\end{figure}

\begin{figure}[!t]
	\centering
	\includegraphics[width=0.49\textwidth]{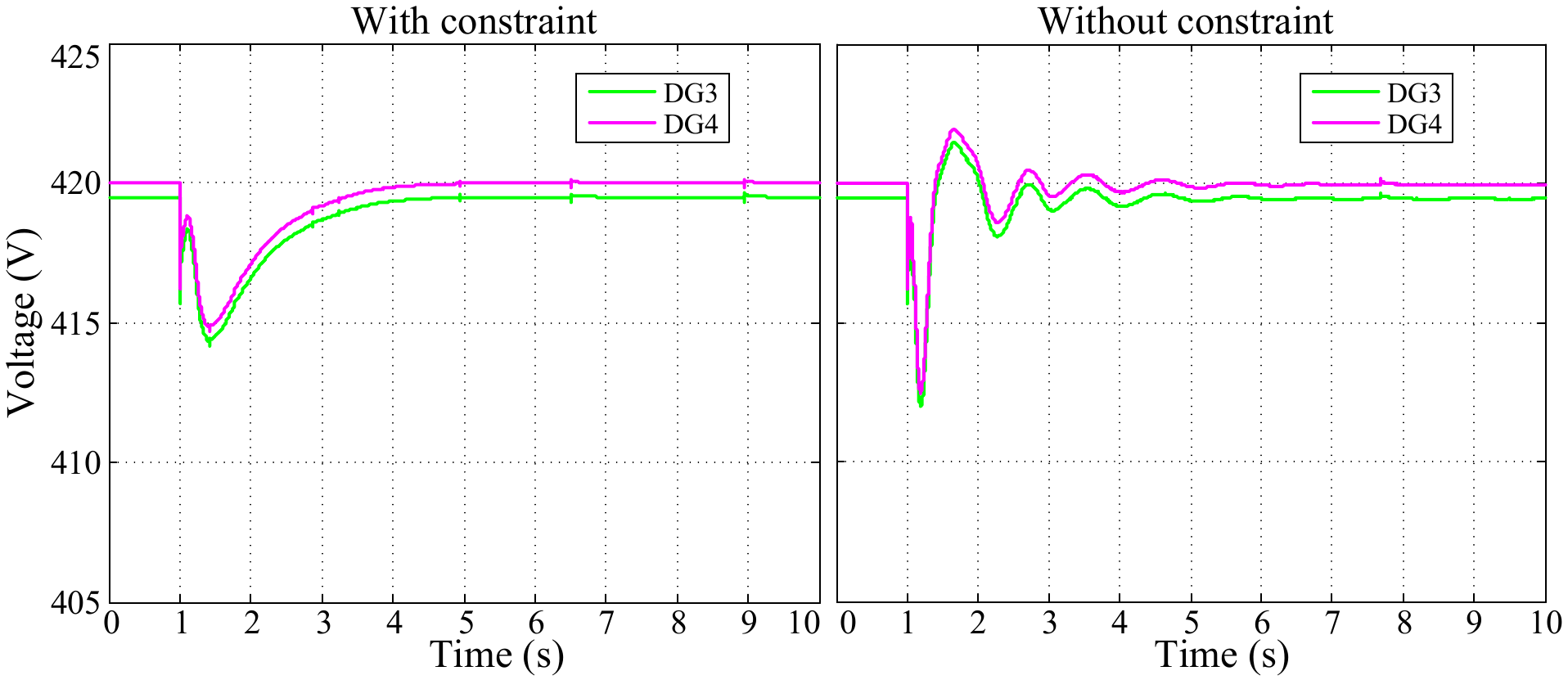}
	\caption{Voltage dynamics  with and without constraints}
	\label{mixed voltage}
\end{figure}

We reset $\overline{p}_i^g=(60, 55, 60, 65, 48, 50)$kW at $t=9$s, where
the power limits of DG2 and DG5 are reduced to $55$kW and $48$kW respectively. This implies that they can be strictly reached in the steady state. This scenario often happens in microgrids since generation limits of renewable resources such as wind turbines and PVs can change rapidly due to uncertainties.
The generations of all DGs with different capacity constraints are given in Fig.\ref{saturation}.

\begin{figure}[!t]
	\centering
	\includegraphics[width=0.35\textwidth]{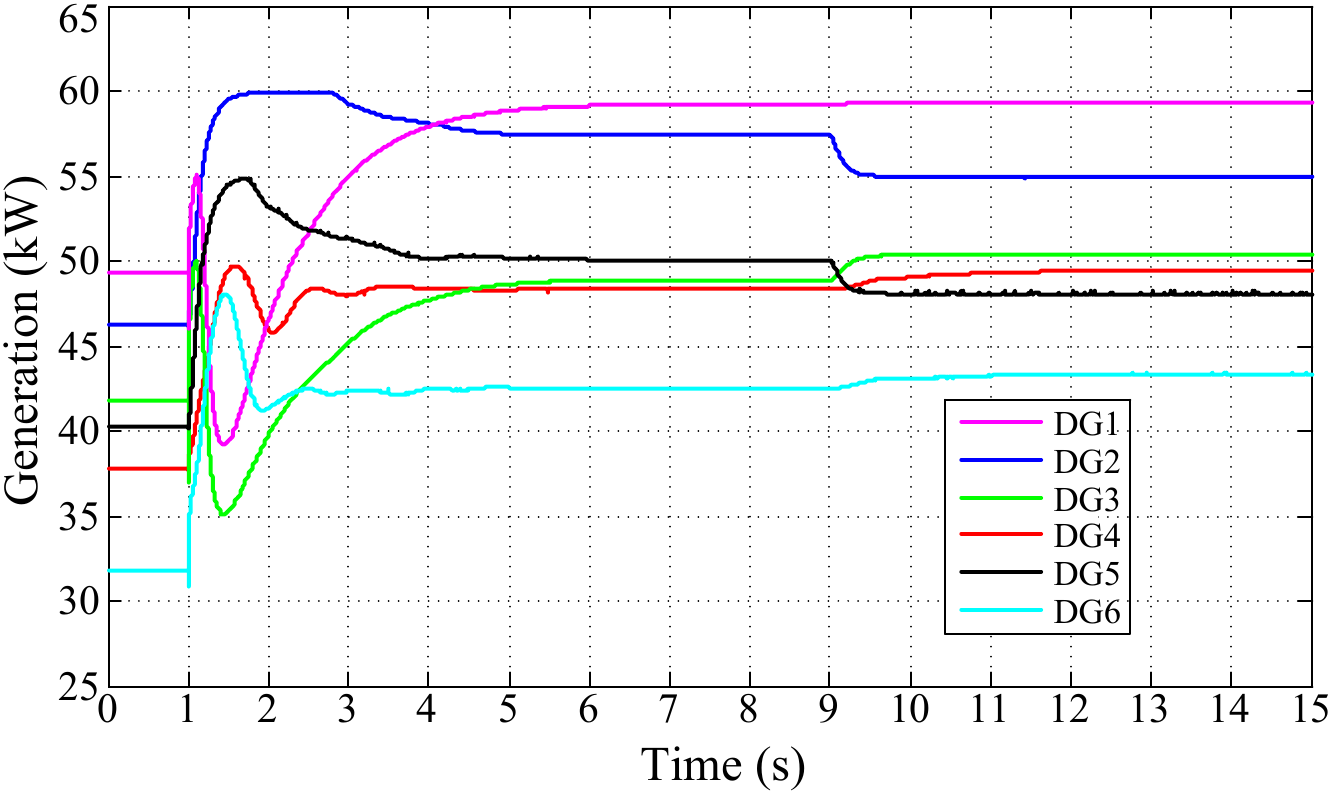}
	\caption{Generation dynamics of different constraints}
	\label{saturation}
\end{figure}

We have checked that generations in Fig.\ref{saturation} are identical with results obtained by CVX. In Fig.\ref{saturation}, it is shown that generations of DG2 and DG5 reduce rapidly to the capacity limits in the new situation.
Other DGs will change their generations to balance the power mismatch in the whole system. This implies our methodology can adapt to disturbances of renewable generations.

%Similarly, we also increase the minimal voltage of DG3 and DG4 to $416$V, which will be reached in the transient process. The voltage dynamics with new and original constraints are illustrated in Fig.\ref{saturation_voltage}. It is shown in Fig.\ref{saturation_voltage} that the voltages are within the limits in the transient even if the limits are very narrow.
%\begin{figure}[!t]
%	\centering
%	\includegraphics[width=0.49\textwidth]{saturation_voltage.pdf}
%	\caption{Voltage Dynamics of Different Constraints}
%	\label{saturation_voltage}
%\end{figure}

%Here, we will investigate the influence of interaction frequency between cyber and physical system. Denote the time inteval of the information exchange between PSCAD and Matlab varies as $\Delta t$, which implies that PSCAD interact with Matlab every $\Delta t$ time. We set $\Delta t=(0.1, 0.12, 0.15, 0.18, 0.2) $ms respectively, and dynamics of generation of DG5 and voltage of DG3 are shown in Fig.\ref{different reading}.
%
%\begin{figure}[!t]
%	\centering
%	\includegraphics[width=0.49\textwidth]{reading_pg_v.pdf}
%	\caption{Generation and voltage dynamics of different time intervals}
%	\label{different reading}
%\end{figure}
%
%From Fig.\ref{different reading}, it is seen that the steady state values of both generation and voltage for different time intervals are identical, which validates that the proposed method can adapt to different interaction time intervals. The only difference is the convergence speed, where large $\Delta t$ always implies slow convergence speed.

\subsection{Plug-n-play Analysis}

In this case, microgrid 6 is switched off at 1s, then it is switched on at 9s. When microgrid 6 is switched off, it has to supply the load demand itself while microgirds 1 to 5 remain connected. Voltage and generation dynamics in the whole process are illustrated in Fig. \ref{plug and play}.

\begin{figure}[!t]
	\centering
	\includegraphics[width=0.49\textwidth]{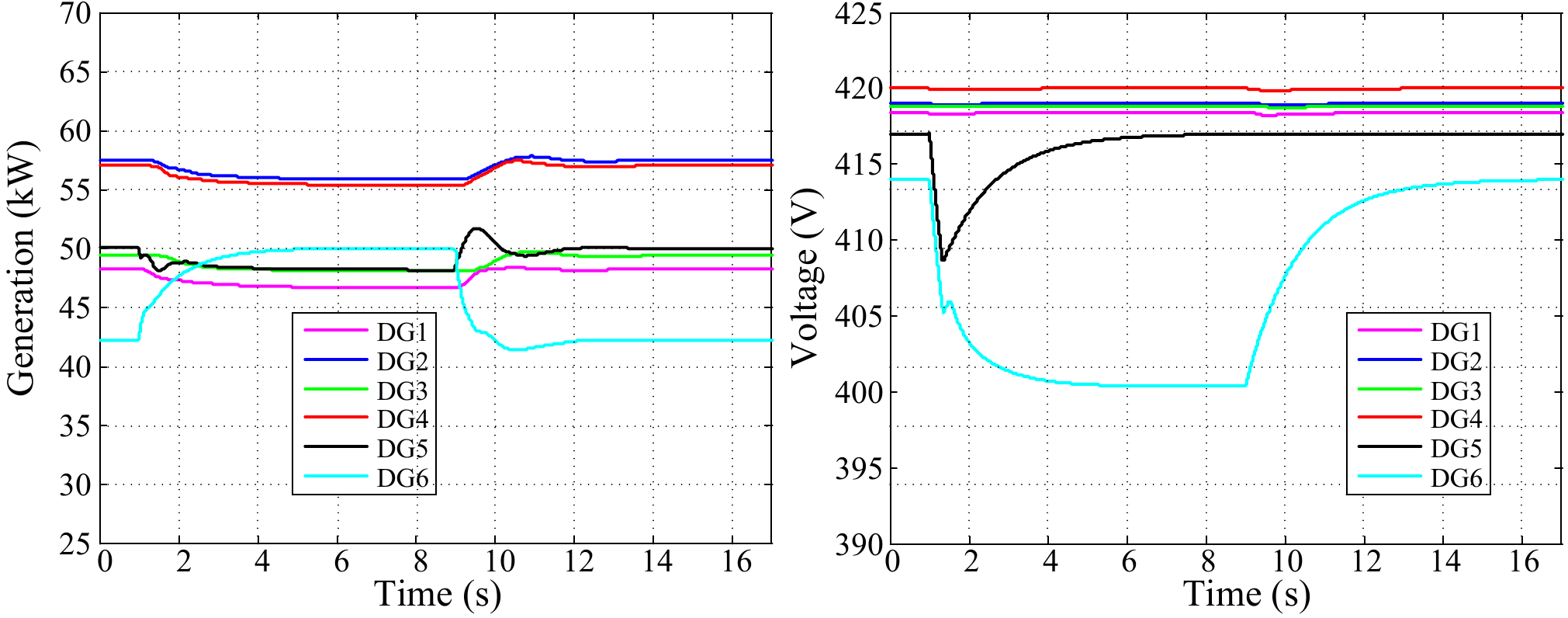}
	\caption{Generation and voltage dynamics of plug and play}
	\label{plug and play}
\end{figure}

It is shown that output of DG6 increases to 50kW to supply the load in microgrid 6 after being switched off. At the same time, the voltage of DG6 reduces to 400V. The result is identical with  that obtained by CVX. In addition, after microgrid 6 is switched on once again, the generations and voltages of all DGs recover to the original values. Moreover, by comparing the voltages in the transient process, it is shown that only the two DGs connected directly with the breaking point are influenced greatly, while other DGs like DG1-DG4 have very moderate transient process.  This validates that our controller can realize plug-n-play.

\section{Conclusion}
This paper addresses the distributed optimal control of stand-alone DC microgrids, where each microgrid may adopt one of the three different control strategies, such as power control, voltage control and droop control. The controller can provide commands for all these strategies, which implies it breaks restriction of various control strategies to achieve optimal operation point. A six-microgrid system based on the microgrid benchmark is utilized to demonstrate the efficacy of our designs. The error of results between proposed method and CVX tool is smaller than $0.6\%$, which validates the accuracy of the proposed approach. Moreover, the commands for power controlled and voltage controlled microgrids satisfy generation limits and voltage limits in both transient process and steady state. This increases the security of DC system. In addition, our controller can adapt to the uncertainties of renewable generations. Finally, the proposed approach can realize the plug-n-play. 

The tie-line limit is not considered in this work since the convex relaxation may be not exact if it is included. In the normal operation, tie line limit in microgrids is often satisfied by planning stage. However, it is also very important when large disturbance happens. In the future research, we will investigate approaches to addressing this problem.
%\section*{Acknowledgment}

\ifCLASSOPTIONcaptionsoff
  \newpage
\fi

\bibliographystyle{Bibliography/IEEEtranTIE}
\bibliography{mybib}

%\begin{thebibliography}{1}
%\end{thebibliography}

%\newpag

%\newpage
%$\ $
\newpage
\appendices

\makeatletter
\@addtoreset{equation}{section}
\@addtoreset{theorem}{section}
\makeatother
\renewcommand{\theequation}{A.\arabic{equation}}
\renewcommand{\thetheorem}{A.\arabic{theorem}}

\section{Proofs of Theorem \ref{uniqueness} and Theorem \ref{Th:same solution}}
\subsection{Proof of Theorem \ref{uniqueness}}

\begin{proof}
	If A2 holds, problem (\ref{initial optimization problem}) is also feasible due to the one-to-one map (\ref{map}). It suffices to prove the uniqueness of the optimal solution of SOCP. Let $x^{1*}=(p^{g1*},v^{1*},W^{1*})$ and $x^{2*}=(p^{g2*},v^{2*},W^{2*})$ be two optimal solutions of SOCP, then we have 
	\begin{align}
	\label{same optimal value}
	\sum\nolimits_{i \in {\cal N}} {f_i(p_i^{g1*}-p_i^d)}=\sum\nolimits_{i \in {\cal N}} {f_i(p_i^{g2*}-p_i^d)}
	\end{align}
	From the proof of Theorem 3 in \cite{Gan:Optimal}, we know 
	\begin{align}
	\frac{v_i^{1*}}{v_i^{2*}}&=\frac{v_k^{1*}}{v_k^{2*}}=\eta,\quad i\sim k\nonumber\\
	W_{ik}^{1*}&= \sqrt{v_i^{1*}v_k^{1*}}=\eta \sqrt{v_i^{2*}v_k^{2*}}=\eta W_{ik}^{2*} \nonumber
	\end{align}
	From (\ref{power flow1}), we have 
	\begin{align}
	p_i^{g1*}- p_i^d &=  \sum\nolimits_{k:k \sim i} (v_i^{1*}-W_{ik}^{1*})/r_{ik} \nonumber\\
	&=  \sum\nolimits_{k:k \sim i} (\eta v_i^{2*}-\eta W_{ik}^{2*})/r_{ik} \nonumber\\
	&=\eta(p_i^{g2*}- p_i^d) \nonumber
	\end{align}
	Since $f_i(p_i^g-p_i^d)$ is strictly increasing, we must have  $\eta=1$, otherwise it contradicts (\ref{same optimal value}). We have $x^{1*}=x^{2*}$, implying the uniqueness of SOCP solution. According to the one-to-one map (\ref{map}), solution of SSOCP is also unique. This completes the proof.
\end{proof}

\subsection{Proof of Theorem \ref{Th:same solution}}
\begin{proof}
	$\Rightarrow$ 1) Suppose $x^{1*} = (p^{g1*}, P^{1*},$ $l^{1*},$ $v^{1*})$ is the optimal solution of (\ref{initial optimization problem}), there exists an unique $\hat p^*$ satisfying (\ref{droop control}). Since two problems have same objective function and constraints except constraint (\ref{droop control}), $(x^{1*},\hat p^{2*})$ is the optimal solution of DSOCP.
	
	$\Rightarrow$ 2) Based on Theorem \ref{uniqueness} and assertion 1) of Theorem \ref{Th:same solution}, this assertion is easy to obtain.
	
	$\Rightarrow$ 3) Since $(p^{g2*},P^{2*},l^{2*},v^{2*})$ is the optimal solution of DSOCP, it also satisfies all the constraints of SSOCP. Moreover, DSOCP and SSOCP have identical objective functions, hence $(p^{g2*},P^{2*},l^{2*},v^{2*})$ is the optimal solution of SSOCP. Due to the uniqueness of optimal solution of SSOCP, we have $(p^{g2*},P^{2*},l^{2*},v^{2*})=(p^{g1*},P^{1*},l^{1*},v^{1*})$. This completes the proof.	
\end{proof}

\makeatletter
\@addtoreset{equation}{section}
\@addtoreset{theorem}{section}
\makeatother
\renewcommand{\theequation}{B.\arabic{equation}}
\renewcommand{\thetheorem}{B.\arabic{theorem}}
\section{Proofs of Lemma \ref{lemma saturation optimal solution} and Theorem \ref{theorem optimality}}
\subsection{Proof of Lemma \ref{lemma saturation optimal solution}}

\begin{proof}
	With assumption A1, A2 and A4, the strong duality holds. $(x_p^*, x_d^*)$ is the primal-dual optimal if and only if it satisfies the KKT conditions.
	
	The Lagrangian of ESOCP is given in (\ref{Lagrangian}).

	\begin{align}
		L &= \sum\nolimits_{i \in {\cal N}} f_i(p_i^{g}-p_i^d)  + \sum\nolimits_{i \in {\cal N}} {\frac{1}{2}{z_i^2}} + \sum\nolimits_{i \in {\cal N}} {\frac{1}{2}{y_i^2}}  \nonumber\\
		&\qquad +\sum\nolimits_{i \in {\cal N}}\epsilon_i(v_i+k_ip^g_i-v_i^*-k_i\hat p_i)  \nonumber \\
		&\qquad + \sum\nolimits_{(i,k) \in {\cal E}} {{\lambda _{ik}}\left( {{P_{ik}} + {P_{ki}} - {r_{ik}}{l_{ik}}} \right)}  \nonumber\\
		&\qquad + \sum\nolimits_{i \in {\cal N}} {{\gamma _{ik}}\left( {{v_i} - {v_k} - {r_{ik}}\left( {{P_{ik}} - {P_{ki}}} \right)} \right)} \nonumber\\
		&\qquad + \sum\nolimits_{(i,k) \in {\cal E}} {{\rho _{ik}}\left( {{{P_{ik}^2}}/{{{v_i}}} - {l_{ik}}} \right)} 
		\label{Lagrangian}
	\end{align}
%	\newcounter{TempEqCnt}
%	\setcounter{TempEqCnt}{\value{equation}}
%	\setcounter{equation}{9}
%	\begin{figure*}[!t]
%		\begin{align}
%		L &= \sum\limits_{i \in {\cal N}} f_i(p_i^{g}-p_i^d)  + \sum\limits_{i \in {\cal N}} {\frac{1}{2}{z_i^2}} + \sum\limits_{i \in {\cal N}} {\frac{1}{2}{y_i^2}}  - \sum\limits_{i \in {\cal N}} {{\mu _i}\left( {p_i^g - p_i^d - \sum\limits_{k:k \in {N_i}} {{P_{ik}}} } \right)}+\sum\limits_{i \in {\cal N}}\epsilon_i(v_i+k_ip^g_i-v_i^*-k_i\hat p_i)  \nonumber \\
%		&\quad + \sum\limits_{ik \in {\cal E}} {{\lambda _{ik}}\left( {{P_{ik}} + {P_{ki}} - {r_{ik}}{l_{ik}}} \right)}  + \sum\limits_{i \in {\cal N}} {{\gamma _{ik}}\left( {{v_i} - {v_k} - {r_{ik}}\left( {{P_{ik}} - {P_{ki}}} \right)} \right)}  + \sum\limits_{ik \in {\cal E}} {{\rho _{ik}}\left( {\frac{{P_{ik}^2}}{{{v_i}}} - {l_{ik}}} \right)} 
%		\label{Lagrangian}
%		\end{align}
%		\hrulefill
%		\vspace*{2pt}
%	\end{figure*}  
%	\setcounter{equation}{\value{TempEqCnt}} 
%	
Based on (\ref{Lagrangian}) we can obtain the KKT conditions
	\begin{subequations}
		\begin{align}
		&{G_i(p_i^{g*}) - {\mu _i^*} +k_i\epsilon_i^* + z_i^*+k_iy_i^* }  \left\{ \begin{array}{lll}
		\geq  0,\  p^{g*}_j = 0\\
		=  0,\  0 < p^{g*}_j < \overline{p}^g_j \\
		\leq  0,\  p^{g*}_j = \overline{p}^g_j \\
		\end{array}  \right.		
		\label{KKT generation}\\
		&y_i^*+ \sum\limits_{k \in {N_i}} {{\gamma _{ik}^*}} +\epsilon_i^*  - \sum\limits_{k \in {N_i}} {{\rho^* _{ik}}\frac{{(P_{ik}^*)^2}}{{(v_i^*)^2}}}  \left\{ \begin{array}{lll}
		\geq  0,\  p^{g*}_j = 0\\
		=  0,\  0 < p^{g*}_j < \overline{p}^g_j \\
		\leq  0,\  p^{g*}_j = \overline{p}^g_j \\
		\end{array}  \right.
		\label{KKT voltage square}\\
		&0 =  - \bigg( {{\mu _i^*} + {\lambda _{ik}^*} - {\gamma _{ik}^*}{r_{ik}} + 2{\rho _{ik}^*}{{{P_{ik}^*}}}/{{{v_i^*}}} - z_i^*} \bigg)
		\label{KKT line power}		\\
		&0 =  - \left( { - {\lambda _{ik}^*}{r_{ik}} - {\rho _{ik}^*} - {\rho _{ki}^*}} \right)
		\label{KKT current square}\\
		&0=k_i\epsilon_i^*+k_iy_i^* \label{KKT hat p}
		\\
		&0 =  -  \bigg( {p_i^{*g} - p_i^d - \sum\limits_{k:k \in {N_i}} {{P_{ik}^*}} } \bigg)\label{KKT mu}\\
		&0 = v_i^*+k_ip^{g*}_i-v_i^*-k_i\hat p_i^*\\	
		&0 =   {{P_{ik}^*} + {P_{ki}^*} - {r_{ik}}{l_{ik}^*}}\\
		&0 =   {{v_i^*} - {v_k^*} - {r_{ik}^*}\left( {{P_{ik}^*} - {P_{ki}^*}} \right)}\label{KKT gamma} \\
		&0 =   \left( {{{(P_{ik}^*)^2}}/{{{v_i^*}}} - {l^*_{ik}}} \right)\rho _{ik}^*,\ \rho^* _{ik}\ge 0 
		\label{KKT rho}
		\end{align}
		$(x_p^*, x_d^*)$ is a primal-dual optimal if and only if it satisfies the KKT conditions. It can be checked that (\ref{KKT generation})  and (\ref{KKT voltage square}) are equivalent to	
		\label{KKT conditions}
	\end{subequations} 
	\begin{align}
	p_i^{g*} &= \left[ {p_i^{g*} - \left( {G_i(p_i^{g*}) - {\mu _i^*} +k_i\epsilon_i^* + z_i^*+k_iy_i^* } \right)} \right]_0^{\overline p_i^g}  \nonumber\\
	v_i^* &=\bigg [v_i^* - \bigg(y_i^*+ \sum\limits_{k \in {N_i}} {{\gamma _{ik}^*}} +\epsilon_i^*  
	- \sum\limits_{k \in {N_i}} {{\rho^* _{ik}}\frac{{(P_{ik}^*)^2}}{{(v_i^*)^2}}}  \bigg)\bigg]_{{\underline{V}_i}^2}^{{\overline{V}_i}^2}\nonumber
	\end{align}
	This completes the proof.	
\end{proof}

\subsection{Proof of Theorem \ref{theorem optimality}}
\begin{proof}
	$\Rightarrow$:
	Suppose $(x_p^*, x_d^*)$  is primal-dual optimal, $(x_p^*, x_d^*)$ satisfies the KKT conditions. It can be obtained directly from (\ref{KKT line power})-(\ref{KKT gamma}) that right sides of dynamics (\ref{line power dynamics})-(\ref{gamma dynamics}) vanish. Right sides of (\ref{generation dynamics}) and (\ref{voltage square dynamics}) vanish due to Lemma \ref{lemma saturation optimal solution}.  From (\ref{KKT rho}) and exactness of convex relaxation, we know 
	\begin{align}
	{{{(P_{ik}^*)^2}}/{{{v^*_i}}} - {l^*_{ik}}}=0,\ \ {\rho }^*_{ik}({{{(P_{ik}^*)^2}}/{{{v^*_i}}} - {l^*_{ik}}})=0\nonumber
	\end{align}
	%	$${\frac{{(P_{ik}^*)^2}}{{{v^*_i}}} - {l^*_{ik}}}=0,$$ 
	%	$${\rho }^*_{ik}\bigg({\frac{{(P_{ik}^*)^2}}{{{v^*_i}}} - {l^*_{ik}}}\bigg)=0 $$	
	Then, the right sides (\ref{rho dynamics}) vanishes. This implies that $(x_p^*, x_d^*)$ is an equilibrium of (\ref{algorithm}). 
	
	\noindent
	$\Leftarrow$: Suppose $(x_p^*, x_d^*)$  is an equilibrium of (\ref{algorithm}), then all the right sides of (\ref{algorithm}) vanish. (\ref{generation dynamics})-(\ref{KKT gamma}) are exactly the KKT conditions (\ref{KKT generation})-(\ref{gamma dynamics}). $\dot\rho_{ik}=0$ implies $\left( {\frac{{(P_{ik}^*)^2}}{{{v_i^*}}} - {l^*_{ik}}} \right)\rho _{ik}^*$ and $\rho^* _{ik}\ge 0$, which is identical to (\ref{KKT rho}). Thus, $(x_p^*, x_d^*)$  is primal-dual optimal. This completes the proof.
\end{proof}

\makeatletter
\@addtoreset{equation}{section}
\@addtoreset{theorem}{section}
\makeatother
\renewcommand{\theequation}{C.\arabic{equation}}
\renewcommand{\thetheorem}{C.\arabic{theorem}}

\section{Proof of Theorem \ref{theorem convergence}}
Define the following  function. 
\begin{align}
\tilde U(x) &=  -   F(x)^T\cdot({H}(x) - x)  \nonumber
\\
&\quad- \frac{1}{2}||{H}(x) - x||_2^2 + \frac{1}{2}||x - x^*||_2^2
\label{Lyapunov function}
\end{align}
From \cite{Fukushima:Equivalent}, we know that $\tilde U(x)\ge0$ and $\tilde U(x)=0$ holds only at any equilibrium point $x^*$.

For any fixed $\sigma_\rho$ , $\tilde U$ is continuously differentiable as $F(x)$ is continuously differentiable in this situation. Moreover, $\tilde U$ is nonincreasing for fixed $\sigma_\rho$, as we will prove in  Lemma \ref{lemma derivative of Lyapunov}. 
It is worthy to note that the index set $\sigma_\rho$ may change sometimes, resulting in discontinuity of $\tilde U$ \cite{Feijer:Stability}. To circumvent such an issue, we slightly modify the definition of $\tilde U$ at the discontinuous points as:
\begin{enumerate}
	\item $U(x) := \tilde{U}(x)$, if $\tilde{U}(x)$ is continuous at $x$;
	\item $U(x) := \limsup\limits_{w\to x} \tilde{U}(w)$, if  $\tilde{U}(x)$ is discontinuous at $x$.
\end{enumerate}
Then $U(x)$ is upper semi-continuous in $x$, and  $ U(x) \geq 0$ on $S$ and $ U(x)=0$ holds only at any equilibrium $x^* = H(x^*)$. 

%As  $U(x)$ is not differentiable for $x$ at discontinuous points, we further define the gradient of $U(x)$ as follows. 
%\begin{enumerate}
%	\item $\frac{\partial U}{\partial x}(x):= \frac{\partial\tilde U}{\partial x}(x)$, if $\tilde U(x)$ is continuous at $x$;
%	\item $\frac{\partial U}{\partial x}(x) := \lim\limits_{w\to x} \left(\frac{\limsup\limits_{w\to x} \tilde U(w) - \tilde U(x)}{w-x}\right)$, if  $\tilde U(x)$ is discontinuous at $x$.
%\end{enumerate}

Note that $U$ is continuous almost everywhere except the switching points. Hence $U(x)$ is \emph{nonpathological} \cite[Definition 3 and 4]{Bacciotti:Nonpathological}. With these definitions and notations above, we can prove  Theorem \ref{theorem convergence}.

To prove Theorem \ref{theorem convergence}, we first start with the following lemma.

\begin{lemma}
	\label{lemma derivative of Lyapunov}
	Suppose A1, A2 and A3 hold. Then
	\begin{enumerate}
		\item $U(x)$ is always decreasing along system (\ref{projection rewritten dynamics}). 
		\item the trajectory $x(t)$ is bounded.
		\item every trajectory $x(t)$ starting from a finite initial state ultimately converges to the largest weakly invariant subset $Z^*$ of $Z^+:=\{\ x\ | \ \dot U(x)= 0\ \}$. 
		\item every $x^*\in Z^*$ is an equilibrium point of \eqref{algorithm}.
	\end{enumerate} 
\end{lemma}
\begin{proof}[Proof of Lemma \ref{lemma derivative of Lyapunov}]
	In light of Theorem 3.2 in \cite{Fukushima:Equivalent}, $U(x)$ is continuously differentiable if $F(x)$ is continuously differentiable. Its gradient is 
	\begin{align}
	\label{gradient Lyapunov}
	\nabla_{x}U(x)&=F(x)-(\nabla_xF(x)-I)(H(x)-x)+x-x^*
	\end{align}
	Then the derivative of $U(x)$ is 
	\begin{align}
	\label{derivative Lyapunov1}
	\dot U(x)=\nabla_{x}^T U(x)\cdot \dot x=\nabla_{x}^T U(x)\cdot (H(x)-x)
	\end{align}
	
	Combining (\ref{gradient Lyapunov}) and (\ref{derivative Lyapunov1}), we have 
	\begin{subequations}
		\begin{align}
		\dot U(x)&=			
		-(H(x)-x)^T\nabla_x F(x)(H(x)-x)
		\nonumber\\
		&\quad+\left<F(x)+H(x)-x, \  (H(x)-x)\right>\nonumber\\
		&\quad+\left<x-x^*, \  (H(x)-x)\right>\nonumber\\
		&=\left<F(x)+H(x)-x, \  H(x)-x^*+x^*-x\right>\nonumber\\
		&\quad+\left<x-x^*, \ H(x)-x\right>\nonumber\\
		&\quad-(H(x)-x)^T\nabla_x F(x)(H(x)-x)\nonumber\\
		&=\left<F(x)+H(x)-x, \  H(x)-x^*\right>\label{dot W a}\\
		&\quad+\left<x-x^*, \ -F(x)\right>\label{dot W b}\\
		&\quad-(H(x)-x)^T\nabla_x F(x)(H(x)-x)\label{dot W c}
		\end{align}
	\end{subequations}
	Next, we will prove that (\ref{dot W a}), (\ref{dot W b}) and (\ref{dot W c}) are all nonpositive. For $\xi$ and $\chi$, the projection has the following property \cite{Fukushima:Equivalent}
	\begin{align}
	\label{projection feature}
	\left<\ \xi-\text{Proj}(\xi)_{S}, \ \chi -\text{Proj}(\xi)_{S}\ \right>  \le 0  \quad \forall \chi \in S \nonumber
	\end{align}		
	Set $\xi=x-F(x)$, $\chi=x^*$, then we have 
	\begin{align}
	\left<F(x)+H(x)-x,\  H(x)-x^*\right>&\le 0 . 
	\end{align}
	This implies that (\ref{dot W a}) is nonpositive.
	
	Write $x_1:=(p_i^g, v_i, P_{ik}, l_{ik}, \hat{p}_i)$ and $x_2:=(\mu _i, \epsilon_i, \lambda _{ik}, \gamma _{ik},$ $\rho _{ik})$, then $L$ is convex in $x_1$ and concave in $x_2$. It can be verified that 
	\begin{align}
	&\left<x-x^*, \ -F(x)\right> =-(x_1-x_1^*)^T\nabla_{x_1}^TL + (x_2-x_2^*)^T\nabla_{x_2}^TL \nonumber\\
	&\qquad \le L(x_1^*,x_2)-L(x_1,x_2) +L(x_1,x_2)-L(x_1,x_2^*)  \nonumber\\
	&\qquad = \underbrace{ L(x_1^*,x_2)-L(x_1^*,x_2^*)}_{\le 0} +\underbrace{L(x_1^*,x_2^*)-L(x_1,x_2^*)}_{\le 0} \nonumber\\
	&\qquad \le 0
	\label{saddle property}
	\end{align}
	This implies that (\ref{dot W b}) is nonpositive.
	
	For (\ref{dot W c}), we have
	\begin{align}
	\label{derivative Lyapunov2}
	-(H(x)-x)^T\nabla_x F(x)(H(x)-x)&=-\dot {x}^T\nabla_x F(x) \dot {x}\nonumber\\
	&=-\dot {x}^TQ \dot {x}\nonumber\\
	&\le 0
	\end{align}
		\newcounter{TempEqCnt2}
	\begin{figure*}[!t]
		\begin{align}
		\label{Q}
		Q=\left[ {\begin{array}{*{20}{c}}
			{\nabla_{p_i^g}G + I_n+K^2}& {-T}& 0 & K & -K^2 &{-I_n} &  K & 0 & 0 & 0 & 0&0
			\\
			{ - T^\text{T}}&{\left[ \frac{{2{\rho _{ik}}}}{{{v_i}}} \right]_d + I_{2m}}&0&D^\text{T}& - T^\text{T}& 0 &(\tilde I^1)^\text{T}&(\tilde I^1)^\text{T}&0&0&\left[{\frac{{2{P_{ik}}}}{{{v_i}}}} \right]_d&0
			\\
			0&0&0&0&0&0&{ -[r_{ik}]_d}&0&0&0&{ - \tilde I^1}&0
			\\
			K&D&0&M_n+I_n&-K&I_n&0&{\tilde I^2}&{ - I_n}&I_n& R&0
			\\
			-K^2&0&0&-K&K^2&0&-K&0&0&0&0&0
			\\
			I_n&{T}&0&0&0&0&0&0&0&0&0&0
			\\
			-K&0&0&-I_n&0&0&K&0&0&0&0&0
			\\
			0&{ - \tilde I^1}&[r_{ik}]_d&0&0&0&0&0&0&0&0&0
			\\
			0&{-\tilde I^{1}}&0&{ - (\tilde I^2)^\text{T}}&0&0&0&0&0&0&0&0
			\\
			0&0&0&I_n&0&0&0&0&0&0&0&0
			\\
			0&0&0&{ - I_n}&0&0&0&0&0&0&0&0
			\\
			0&\left[{ - \frac{{2{P_{ik}}}}{{{v_i}}}} \right]_d& (\tilde I^1)^\text{T} &-R^\text{T}&0&0&0&0&0&0&0&0
			\end{array}} \right]
		\end{align}		
		\hrulefill
		\vspace*{2pt}
	\end{figure*}  
	\setcounter{equation}{\value{TempEqCnt2}} 
	where $Q$ is given in (\ref{Q}) with 
	\begin{align}
%	&A=\text{diag}(a_i, i\in\mathcal{N} )\nonumber\\
	&M_n=\text{diag}\left({\sum\nolimits_{k \in {N_i}} {\frac{{2{\rho _{ik}}P_{ik}^2}}{{v_i^3}}} }\right) \nonumber\\		
	&{T}_{n\times 2m}(s,t)=   \left\{ \begin{array}{l}
	1, \ \text {if} \ (s,t+1)\in \mathcal{E} \ \text{or}\ (s,t-m+1)\in \mathcal{E}\\
	0, \ \ \ \ \ \ \text {otherwise}
	\end{array} \right.    \nonumber\\
	&{D}_{n\times 2m}(s,t)=   \left\{ \begin{array}{l}
	- \frac{{2{\rho _{s,t+1}}{P_{s,t+1}}}}{{v_s^2}}, \ \text {if} \ (s,t+1)\in \mathcal{E} \ \\
	\ \ \ \ \ \ \ \ \ \ \ \ \  \text{or}\ (s,t-m+1)\in \mathcal{E}\\
	0, \ \ \ \ \ \text {otherwise}
	\end{array} \right.    \nonumber\\		
	&{R}_{n\times 2m}(s,t)=   \left\{ \begin{array}{l}
	{ - \frac{{P_{s,t+1}^2}}{{v_s^2}}}, \ \text {if} \ (s,t+1)\in \mathcal{E} \ \\
	\ \ \ \ \ \ \ \ \ \ \ \text{or}\ (s,t-m+1)\in \mathcal{E}\\
	0, \ \ \ \ \  \text {otherwise}
	\end{array} \right.    \nonumber\\		
	&{\tilde I^1}_{m\times 2m}(s,t)=   \left\{ \begin{array}{l}
	1, \ \text {if} \ (s,t+1)\in \mathcal{E} \ \text{or}\ (s,t-m+1)\in \mathcal{E}\\
	0, \ \ \ \ \ \ \text {otherwise}
	\end{array} \right.    \nonumber\\		
	&{\tilde I^2}_{m\times 2m}(s,t)=   \left\{ \begin{array}{l}
	{  1},\ \ \ \text {if} \ (s,t+1)\in \mathcal{E}\ \text{and}\ s\le t \\
	{  -1}, \ \text {if} \ (s,t-m+1)\in \mathcal{E}\ \text{and}\ s>t-m \\
	0, \ \ \ \ \ \ \text {otherwise}
	\end{array} \right.    \nonumber		
	\end{align}
	$Q$ is a semi-definite positive matrix. $I$ is a identity matrix, the subscript implies its dimension. $[c_i]_d$ denotes the diagonal matrix composed of $c_i$ with proper dimensions. Moreover, $Q$ can be divided into two matrices, one of which is skew-symmetric and the other is positive symmetric.

	Note that the index set $\sigma_\rho$ may change during the decreasing of $U$.
	We have the following observations:
	\begin{itemize}
		\item The set $\sigma_\rho$ is reduced, which only happens when ${{{P_{ik}^2}}/{{{v_i}}} - {l_{ik}}}$ goes through zero, from negative to positive. Hence an extra term will be added to $U$. As this term is initially zero, there is no discontinuity of $U$ in this case.
		\item The set $\sigma_\rho$ is enlarged when $\rho_{ik}$ goes to zero from positive while ${{{P_{ik}^2}}/{{{v_i}}} - {l_{ik}}}\le 0$. Here $U$ will lose a positive term $\left(\frac{{P_{ik}^2}}{{{v_i}}} - {l_{ik}}\right)^2/2$, causing discontinuity.
	\end{itemize}
	
	Hence, $U$ keeps decreasing even when $\sigma_\rho$ changes, which implies 1) of Lemma \ref{lemma derivative of Lyapunov}.
	In addition, note that \cite[Theorem 3.1]{Fukushima:Equivalent} proves that $ - F(x) ^T \left( H(x)-x \right)-\frac{1}{2} ||H(x)-x||^2_2 \geq 0$. 	
	Therefore, we have
	$$\frac{1}{2}||x-x^*||_2^2\le U(t)\le U(0)$$
	which implies that $x(t)$ is bounded. Then, 2) of Lemma \ref{lemma derivative of Lyapunov} holds.
	
	Given an initial point $x(0)$ there is a compact set $\Omega_0 := \Omega(x(0)) \subset S$ such that $x(t)\in\Omega_0$ 
	for $t\geq 0$ and $\dot U(x) \leq 0$ in $\Omega_0$. 
	
	In addition, $U$ is radially unbounded and positively definite except at equilibrium. As $U$ and $\dot U$ are nonpathological, we conclude that any trajectory $x(t)$ starting from $\Omega_0$ converges to the largest weakly invariant subset $Z^*$ contained in $Z^+=\{\ x\in \Omega_0\ |\ \dot U(x) =0\ \}$ \cite[Proposition 3]{Bacciotti:Nonpathological}, proving  the third assertion.
	
	Now, we will prove the last assertion of Lemma \ref{lemma derivative of Lyapunov}. To satisfy $\dot U(x)= 0$, both terms in (\ref{saddle property}) have to be zero, implying that 
	$$L(x_1^*,x_2)\equiv L(x_1^*,x_2^*)$$
	must hold in $Z^+$. Differentiating with respect to $t$ gives 
	\setcounter{equation}{8}
	\begin{align}
	\left(\frac{\partial}{\partial x_2}L(x_1^*,x_2(t))\right)^T\cdot\dot x_2(t)=0 =\dot x_2(t)^T\dot x_2(t)
	\end{align}
	The second equality holds due to (\ref{mu dynamics})-(\ref{rho dynamics}). Then, we can conclude $\dot x_2(t)=0$ due to the boundedness of $x(t)$, which implies that $\mu _i, \epsilon_i, \lambda _{ik}, \gamma _{ik},$ $\rho _{ik}$ are constants and $y_i=z_i=0$ in $Z^+$. We can obtain $\dot l_{ik}=\dot{\hat{p}}_i=0$ from (\ref{current square dynamics}), (\ref{hat p dynamics}) as well as the boundedness of $x(t)$. 
	%	 Now, we will prove $\dot P_{ik}=0$. From (\ref{line power dynamics}), only $2\rho_{ik}\frac{P_{ik}}{v_i}$ is indeterminate. As $\dot \rho_{ik}\equiv0$, $\rho_{ik}$ can be categoried into two situations: 1) $\rho_{ik}=0$, 2) $\rho_{ik}>0$. If $\rho_{ik}=0$, the right side of $\dot P_{ik}$ is constant. In this case, $\dot P_{ik}=0$ holds due to the boundedness of $P_{ik}$. If $\rho_{ik}>0$, we have ${\frac{{P_{ik}^2}}{{{v_i}}} = {l_{ik}}}$, equally ${\frac{{P_{ik}}}{{{v_i}}} = \frac{l_{ik}}{P_{ik}}}$
	
	Combining (\ref{derivative Lyapunov2}) and (\ref{Q}), we have 
	\begin{align}
	\label{derivative Lyapunov3}
	\dot U(x)&\le-\dot {x}^TQ \dot {x}\nonumber\\
	&=-\sum\limits_{i \in {\cal N}} {({\dot p_i^g})^T\cdot\nabla_{p^g}^2f\cdot\dot p_i^g}  - \sum\limits_{i \in {\cal N}} { {{{\left( {\dot p_i^g - \sum\limits_{k \in {N_i}} {{\dot P}_{ik}}} \right)}^2}} }  \nonumber\\
	&\quad- \sum\limits_{i \in {\cal N}} { {\frac{{2{\rho _{ik}}}}{{{v_i}}}{{\left( {{{\dot P}_{ik}} - \sum\limits_{k \in {N_i}} \frac{{\left| {{P_{ik}}} \right|}}{{{v_i}}}{{\dot v}_i}} \right)}^2}} }\nonumber\\
	&\quad-\sum\limits_{i \in {\cal N}} {(\dot v_i+k_i\dot p^g_i-k_i\dot{\hat p}_i)^2}  
	\end{align}
	We can directly get $\dot p_i^g=0$ due to the A1. From $\dot v_i+k_i\dot p^g_i-k_i\dot{\hat p}_i=0$ and $\dot p_i^g=\dot{\hat{p}}_i=0$, we have $\dot v_i=0$. If $\rho_{ik}=0$, then $\dot{P}_{ik}$ is a constant, implying $\dot{P}_{ik}=0$. If $\rho_{ik}>0$, then ${\frac{{P_{ik}^2}}{{{v_i}}} = {l_{ik}}}$, implying $P_{ik}$ a constant. Thus, $\dot{P}_{ik}=0$ always holds.
	Consequently, we have that $\dot x(t)=0$ in $Z^*$, which is the last assertion of Lemma \ref{lemma derivative of Lyapunov}.
\end{proof}

%We can further prove that only the equilibrium point can exist in $Z^*$. 

\begin{proof}[Proof of Theorem \ref{theorem convergence}]	
	Fix any initial state $x(0)$ and consider the trajectory $(x(t), t\geq 0)$ of (\ref{projection rewritten dynamics}). As mentioned in the proof of Lemma \ref{lemma derivative of Lyapunov}, $x(t)$ stays entirely in a compact set $\Omega_0$. Hence there exists an infinite sequence of time instants ${t_k}$ such that $x(t_k)\to \hat {x}^*$ as $t_k\to\infty$, 
	for some $\hat x^* \in Z^*$. The 4) in Lemma \ref{lemma derivative of Lyapunov} guarantees that $\hat x^*$ is an equilibrium point of the (\ref{projection rewritten dynamics}), and hence $\hat x^*=H(\hat{x}^*)$.
	Thus, using this specific equilibrium point $\hat {x}^*$ in the definition of $U$, we have 
	\begin{align}
	U^* =\lim\limits_{t\to \infty} U(x(t)) &= \lim\limits_{t_k\to \infty} U(x(t_k)) \nonumber\\
	&\qquad=\lim\limits_{x(t_k) \to \hat x^*} U\big( x(t_k)\big) = U(\hat x^*) = 0 \nonumber
	\end{align}
	Here, the first equality uses the  fact that
	$U(t)$ is nonincreasing in $t$; the second equality uses the fact that
	$t_k$ is the infinite sequence of $t$; the third equality uses the fact that $x(t)$ is
	absolutely continuous in $t$; the fourth equality is due to the upper semi-continuity of $U(x)$, and the last equality holds as $\hat {x}^*$ is an equilibrium point of $U(x)$. 
	
	The quadratic term $(x-\hat x^*)^T(x-\hat x^*)$
	in $U(x)$ then implies that $x(t)\to \hat {x}^*$ as $t\to \infty$, which completes the proof. 
\end{proof}

%\section{}

\end{document}